\documentclass[9pt,shortpaper,twoside,web]{ieeecolor}
\usepackage{generic}
\usepackage{amssymb}
\usepackage{amsmath}
\usepackage{color}
\newtheorem{definition}{Definition}
\newtheorem{assumption}{Assumption}
\newtheorem{theorem}{Theorem}
\newtheorem{remark}{Remark}
\newtheorem{lemma}{Lemma}

\newtheorem{corollary}{Corollary}
\usepackage{bm}
\usepackage{caption}
\usepackage{cite}
\usepackage{graphicx}
\usepackage{epstopdf}
\usepackage{cite}
\usepackage{algorithm,algpseudocode,float}
\usepackage{lipsum}
\usepackage{tabularx}
\makeatletter

\def\BibTeX{{\rm B\kern-.05em{\sc i\kern-.025em b}\kern-.08em
    T\kern-.1667em\lower.7ex\hbox{E}\kern-.125emX}}
\begin{document}

\title{\LARGE \bf
	An Efficient Distributed Nash Equilibrium Seeking with Compressed and Event-triggered Communication}
\author{Xiaomeng Chen$^{1}$, \thanks{$^{1}$X. Chen, W. Huo and L. Shi are with the Department of Electronic and Computer Engineering, Hong Kong University of Science and Technology, Clear Water Bay, Kowloon, Hong Kong (email: \{xchendu,whuoaa,eesling\}@ust.hk).} Wei Huo$^{1}$, Yuchi Wu$^{2}$ \thanks{$^{3}$ Y. Wu is with the School of Mechatronic Engineering and Automation, Shanghai University, Shanghai, China. (email: wuyuchi1992@gmail.com).}, Subhrakanti Dey$^{3}$, \thanks{$^{3}$ S. Dey is with the Department of Electrical Engineering, Uppsala University, Box 65, 751 03 Uppsala, Sweden. (email: subhrakanti.dey@angstrom.uu.se).} and Ling Shi$^{1}$
}

\maketitle

\begin{abstract}
Distributed Nash equilibrium (NE) seeking problems for networked games have been widely investigated in recent years.  Despite the increasing attention, communication expenditure is becoming a major bottleneck for scaling up distributed approaches within limited communication bandwidth between agents.  To reduce communication cost, an efficient distributed NE seeking (ETC-DNES) algorithm is proposed to obtain an NE for games over directed graphs, where the communication efficiency is improved by event-triggered exchanges of compressed information among neighbors. ETC-DNES saves communication costs in both transmitted bits and rounds of communication. Furthermore, our method only requires the row-stochastic property of the adjacency matrix, unlike previous approaches that hinged on doubly-stochastic communication matrices. We provide convergence guarantees for ETC-DNES on games with restricted strongly monotone mappings and testify its efficiency with no sacrifice on the accuracy. The algorithm and analysis are extended to a compressed algorithm with stochastic event-triggered mechanism (SETC-DNES).  In SETC-DNES, we introduce a random variable in the triggering condition to further enhance algorithm efficiency. We demonstrate that SETC-DNES guarantees linear convergence to the  NE while achieving even greater reductions in communication costs compared to ETC-DNES. Finally, numerical simulations illustrate the effectiveness of the proposed algorithms. 
\end{abstract}


\section{INTRODUCTION}

\IEEEPARstart{N}{oncooperative} games have been studied extensively  primarily due to its wide applications  in a large variety of fields, such as congestion control in traffic networks \cite{ma2011decentralized}, charging/discharging of electric vechicles \cite{grammatico2017dynamic} and demand-side management in smart grids \cite{saad2012game}.
As an important issue, Nash equilibrium (NE) seeking has gained ever-increasing attention with the emergence of multi-agent systems. Conceptually, NE represents a proposed solution  in multi-player non-cooperative games, where selfish players aim to minimize their individual cost functions by making decisions based on the actions of others. 

A large body of works on NE seeking  have been reported in the recent literature (see e.g.\cite{yu2017distributed,belgioioso2018projected,shamma2005dynamic} and references therein). Most  of them assumed that each player has access to the decisions of all other players, necessitating a central coordinator to disseminate information across the network. However, this full-decision information setting can be impractical in certain scenarios \cite{ghaderi2014opinion}. Hence, distributed NE seeking algorithms have gained considerable interest recently, where each agent is required to transmit augmented states, which are combinations of real states and estimates at each iteration. This communication overhead can be particularly challenging in games with numerous agents or in systems with limited communication capacities, such as underwater vehicles and low-cost unmanned aerial vehicles. Therefore, reducing communication costs while ensuring convergence to NE in large-scale games is of paramount importance.

	To alleviate  the communication burden in distributed networks, a large number of communication-efficient approaches have emerged. The key idea is to lessen the amount of communication at each iteration. For example, a compressor is usually adopted for agents to send compressed information that is encoded with fewer number of bits. Recently, various  information compression methods, such as quantization and sparsification, have been utilized in decentralized networks for distributed optimization and learning \cite{liao2022compressed,kovalev2021linearly,liu2021linear,beznosikov2020biased,zhang2023innovation}. In the context of  non-cooperative games,  Nekouei et al. \cite{nekouei2016performance}   investigated the impact of quantized communication on the behavior of NE seeking algorithms over fully connected communication graphs. A continuous-time distributed NE seeking algorithm with finite communication bandwidth is proposed by Chen et al. \cite{chen2022distributed}, where a specific quantizer is used. In our previous work, we  developed a discrete-time distributed NE seeking algorithm (C-DNES) that employed general compressors to compress information before transmission \cite{chen2022linear}. 
	
	While compression effectively reduces the amount of information transmitted, it does not  reduce the number of communication rounds. In some instances, the transmitted state may exhibit minimal changes between iterations, resulting in  a waste of communication resources. To address this issue, event-triggered mechanisms have been introduced, drawing inspiration from distributed control \cite{he2019adaptive, ge2021dynamic} and extending to distributed optimization \cite{liu2019distributed, cao2020decentralized}. The core idea behind event-triggered scheme is that information is transmitted only when an event occurs. {\color{black}As for distributed games,  Shi et al. \cite{shi2019distributed} introduced an edge-based event-triggered law for undirected graphs, while Zhang et al. \cite{zhang2021distributed} successfully implemented dynamic event-triggered mechanisms in directed networks to enhance communication efficiency. Furthermore, Xu et al. \cite{xu2021event} developed an event-triggered generalized NE seeking algorithm for distributed games with coupling constraints.} While previous works primarily focused on deterministic event-triggered schemes, recent research has demonstrated the efficiency of stochastic event-triggered mechanisms in reducing communication rounds for distributed systems \cite{han2015stochastic}.  {\color{black}Building on this, Tsang et al. \cite{tsang2022stochastic} extended deterministic event triggers to stochastic versions for distributed optimization problems, and Huo et al. \cite{huo2023distributed}  applied the stochastic event-triggered method to distributed non-cooperative games in continuous time, proving convergence to the exact NE in a mean square sense.} However, these methodologies do not address the data size transmitted at each iteration. {\color{black} Recently, Singh et al. \cite{singh2022sparq}  devised a communication-efficient algorithm that merged information compression with event-triggered communication for distributed optimization problems. This approach, however, faces distinct challenges when adapted to the context of distributed non-cooperative games, where each player's decision critically depends on the actions of others. This dependency necessitates not only the ongoing estimation of other players' actions but also the application of complex game theory techniques such as the variational inequality approach to reach to the NE.

Inspired by the aforementioned discussions, we  develop a novel communication-efficient NE seeking algorithm for distributed non-cooperative games over unbalanced directed graphs, which innovatively combines compression with event-triggered communication. Traditional event-triggered mechanisms based on local variables \cite{shi2019distributed, zhang2021distributed, xu2021event} do not account for compression errors, which can severely degrade performance. Our approach introduces novel deterministic and stochastic event-triggered mechanisms where the triggering conditions are intelligently designed around the compressed values.  Furthermore, we employ constant stepsizes to accelerate convergence. This choice, while enhancing convergence speed, introduces additional complexities in the algorithm's design and analysis due to the persistent nature of compression and event-triggered errors, which, if not properly managed, could lead to divergence. To ensure accurate convergence under constant stepsizes, we adopt the difference compression framework from our previous work \cite{chen2022linear} and meticulously design the triggering conditions to establish linear convergence.}

The main contributions of this paper are summarized as follows:

\begin{enumerate}

\item We  combine the compression approach with    event-triggered  strategy to design  a  compressed event-triggered distributed NE seeking algorithm, ETC-DNES, where agents initiate a compressed information transmission controlled by a locally  triggering event.
Compared with the compression approach in \cite{chen2022distributed,chen2022linear} and the event-triggered approach in \cite{shi2019distributed,zhang2021distributed,xu2021event}, our proposed algorithm can 
 save communication bits and rounds simultaneously, which further reduce the communication cost. 
\item We analyze the  convergence  performance of the proposed   algorithm for  games with restricted strongly monotone mappings. The results show that  ETC-DNES   achieves exact  convergence asymptotically under the diminishing triggering  condition sequence (\textbf{Theorem 1}). We also prove that when the triggering sequence is bounded by an exponential decay function,  linear convergence is established (\textbf{Corollary 1}).
\item We further extend the event-triggered mechanism  to a stochastic one, where a random variable is introduced, and propose SETC-DENS to improve the algorithm efficiency. We show that it can converge to the  NE with linear convergent rate (\textbf{Theorem 2}).
\item    The proposed strategies accommodate a broad  class of compressors with bounded relative compression error, encompassing commonly used compressors including unbiased and biase but contractive compressors. Moreover, unlike previous works \cite{shi2019distributed,xu2021event,huo2023distributed} that require undirected communication graph,  our proposed strategy for communication-efficient NE seeking applies to both  undirected and directed  graphs. 
\item In our simulation examples, we investigate various options for the event-triggered conditions. The results demonstrate the remarkable efficiency of both ETC-DNES and SETC-DNES in  reducing the number of transmitted data. Notably, these algorithms achieve reductions of up to a hundredfold in the number of total transmitted bits. 

\end{enumerate}

{\color{black}Different from  the recent compressed and event-triggered algorithms  in \cite{singh2022sparq}, which focus on distributed optimization problems over  undirected graphs with doubly-stochastic communication matrices, this paper tackles distributed non-cooperative games over potentially unbalanced directed communication graphs. Additionally, our approach is compatible with a broader range of compressor models, supporting compression ratios greater than $1$, and achieves linear convergence, thereby ensuring faster convergence rates. We also introduce a stochastic event-triggered mechanism to further enhance communication efficiency.}

 \textit{Notations:}  $\bm{1}\in\mathbb{R}^n$  represents the column vector with each entry given by $1$. We denote $\bm{e}_i \in\mathbb{R}^n$ as the  unit vector which takes zero except for the $i$-th entry that equals to $1$. $\mathbb{R}_{++}$ denotes the set of all positive real numbers. The spectral radius of matrix $\mathbf{A}$ is denoted by $\rho(\mathbf{A})$. The smallest nonzero eigenvalue of a positive semidefinite matrix $ \mathbf{M} \succeq  \mathbf{0}$ is denoted as $\tilde \lambda_{\min}(\mathbf{M})$. Given a vector $\mathbf{x}$, we denote its $i$-th
element by $[\bm{x}]_{i}$. Given a matrix $\mathbf{A}$, we denote its
element in the $i$-th row and $j$-th column by $[\bm{A}]_{ij}$. Let sgn($\cdot$) and $|\cdot|$ be the element-wise sign function and absolute value function, respectively. We denote  $||\mathbf{x}||$ and $||\mathbf{X}||_{\text{F}}$ respectively as $l_2$ norm of vector $\mathbf{x}$ and the Frobenius norm of matrix $\mathbf{X}$. 
 We denote  by $\mathbf{A} \odot \mathbf{B}$, the Hadamard product of two matrices, $\mathbf{A}$ and $\mathbf{B}$. In addition, $\bm{x} \preceq \bm{y}$ is denoted as component-wise inequality between vectors $\bm{x}$ and $\bm{y}$.  For a matrix $\mathbf{A} \in \mathbb{R}^{n\times n}$, denote $\text{diag}(\mathbf{A})$ as its diagonal vector, i.e., $\text{diag}(\mathbf{A})=[[\bm{A}]_{11},\ldots, [\bm{A}]_{nn}]^\top$. For a vector $\mathbf{a} \in \mathbb{R}^{n}$, denote $\text{Diag}(\mathbf{a})$ as the diagonal matrix with the vector $\mathbf{a}$ on its diagonal. A matrix $\mathbf{A} \in \mathbb{R}^{n\times n}$ is consensual if it has equal row vectors. $U(a,b)$ denotes the uniform distribution over the interval $(a,b).$ 

\section{PRELIMINARIES AND PROBLEM STATEMENT}
\subsection{Network Model}
We consider a group of agents  communicating  over a directed graph $G\triangleq (\mathcal{V}, \mathcal{E})$, where $\mathcal{V}=\{1,2,3\ldots, n\}$ denotes the agent set and $\mathcal{E}\subset \mathcal{V} \times \mathcal{V}$ denotes the edge set, respectively. A communication link from agent $j$ to agent $i$ is denoted by $(j,i)\in \mathcal{ E}$. The agents who can  send messages to node $i$ are denoted as in-neighbours of agent $i$ and the set of these agents is denoted as $\mathcal{N}_i^{-}=\{j\in \mathcal{V}\mid (j,i)\in \mathcal{E}\}$. Similarly,  the agents who can  receive messages from node $i$ are denoted as out-neighbours of agent $i$ and the set of these agents is denoted as $\mathcal{N}_i^{+}=\{j\in \mathcal{V}\mid (i,j)\in \mathcal{E}\}$.  The adjacency matrix of the graph is denoted as $W=[w_{ij}]_{n\times n}$ with $w_{ij}>0$ if $(i,j)\in \mathcal{ E}$ or $i=j$, and $w_{ij}=0$ otherwise. The graph $G$ is called strongly connected if there exists at least
one directed path from any agent $i$ to any agent $j$ in the graph with $i\neq j$. 
\begin{assumption}\label{asp3}
	The directed graph $G\triangleq (\mathcal{V}, \mathcal{E})$ is strongly connected. Moreover, the adjacency matrix $W$ associated with $G$ is row-stochastic, i.e., $\sum_{l=1}^n w_{il}=1, \forall i \in \mathcal{V}$.
		\end{assumption}

\subsection{Compressors}
A stochastic compressor $\mathcal{C}(\mathbf{x}, \xi, C):  \mathbb{R}^d \times \mathcal{Z} \times \mathbb{R}^+ \rightarrow \mathbb{R}^d$ is a mapping that converts  a $\mathbb{R}^d$-valued signal $\mathbf{x}$ to a compressed one, where $C\geq 0$ is the  desired compression ratio and $\mathbf{\xi}$ is a  random variable with range $\mathcal{Z}$. Note that the realizations $\xi$ of the compressor $\mathcal{C}(.)$ are independent  among different agents and time steps. The notation $\mathbb{E}_{\xi}$ is used to  denote  the expectation over the  inherent randomness  in the stochastic compressor. Hereafter, we drop $\xi, C$ and write $\mathcal{C}_{i,k}(\mathbf{x})$ for notational simplicity.

\begin{assumption}\label{c3}
	For any agent $i \in \mathcal{V}$ and any iteration $k\ge0$, given all the recorded history of data on the network $\mathcal{H}_k$,
 the compression operator $\mathcal{C}_{i,k}$ satisfies
\begin{equation}\label{asp2}
	\mathbb{E}_{\xi}[||\mathcal{C}_{i,k}(\mathbf{x} )-\mathbf{x}||^2|\mathcal{H}_k]\leq C||\mathbf{x}||^2,\qquad \forall \mathbf{x}\in\mathbb{R}^d,
\end{equation}
where $C\geq 0$ and the r-scaling of $\mathcal{C}_{i,k}$ satisfies
\begin{equation}
	\mathbb{E}_{\xi}[||\mathcal{C}_{i,k}(\mathbf{x})/r-\mathbf{x}||^2|\mathcal{H}_k]\leq (1-\delta)||\mathbf{x}||^2,\qquad \forall \mathbf{x}\in\mathbb{R}^d, 
\end{equation}
where   $\delta \in (0,1]$ and $r>0$.
\end{assumption}

The above class of compressors incorporate the widely used unbiased compressors and biased but contractive compressors, such as stochastic quantization \cite{liu2021linear} and Top-$k$ sparsification \cite{beznosikov2020biased}. They also cover biased and non-contractive compressors like the norm-sign compressor. In other words, Assumption   \ref{c3} is weaker than various commonly used assumptions for compressors in the literature.

\subsection{Problem Formulation}
Consider a  non-cooperative game  in a multi-agent system of $n$ agents, where 
each agent has an unconstrained action set $\Omega_i=\mathbb{R}$. Without loss of generality, we assume that  each agent's decision variable is a scalar. Let $J_i$ denote  the cost function of the agent $i$. Then, the game is denoted as $\Gamma(n, \{J_{i}\}, \{\Omega_i=\mathbb{R}\})$.   The goal of each agent $i\in \mathcal{V}$ is to minimize its objective function $J_i(x_i,\mathbf{x_{-i}})$, which depends on both the local variable $x_i$ and the decision variables of the other agents $\mathbf{x_{-i}}$. The concept of NE is given below. 
\begin{definition}
	A vector $\mathbf{x}^\star=[x_1^\star,x_2^\star,\ldots,x_n^\star]^\top $ is an NE if for any $i \in \mathcal{V}$ and $x_i \in\mathbb{R}$,
$
		J_i(x_i^\star, \mathbf{x_{-i}}^\star)\leq J_i(x_i, \mathbf{x_{-i}}^\star).
$
\end{definition}

We then make the following assumptions with respect to game $\Gamma$. 
\begin{assumption}\label{cvx}
For  $i \in\mathcal{V}$, the  cost function $J_i(x_i,\mathbf{x_{-i}})$ is strongly convex and continuously differentiable in $x_i$ for each fixed $\mathbf{x_{-i}}$. 
\end{assumption}


\begin{definition}
	The mapping $\mathbf{F}: \mathbb{R}^n \rightarrow \mathbb{R}^n$, referred to be the game mapping of $\Gamma(n, \{J_{i}\}, \{\Omega_i \})$ is denoted as
\begin{equation}
	\mathbf{F(x)}\triangleq [\nabla_1 J_1(x_1,x_{-1}),\ldots, \nabla_n J_n(x_n,x_{-n})]^\top,
\end{equation} 
where $\nabla_i J_i(x_i,\mathbf{x_{-i}})=\frac{\partial J_i(x_i,\mathbf{x_{-i}})}{\partial x_i},\forall i \in\mathcal{V}$. 
\end{definition}

\begin{assumption}\label{mono}
	The game mapping $\mathbf{F(x)}$ is restricted strongly monotone  to any NE $\mathbf{x}^\star $ with constant $\mu_r>0$, i.e.,
	$$\langle \mathbf{F(x)}-\mathbf{F(\mathbf{x}^\star)},\mathbf{x}-\mathbf{\mathbf{x}^\star}  \rangle \geq \mu_r ||\mathbf{x-\mathbf{x}^\star}||^2, \forall \mathbf{x} \in\mathbb{R}^n.$$
		 
\end{assumption}

\begin{assumption}\label{asp2}
	Each function $\nabla_i J_i(x_i,x_{-i}) $ is uniformly Lipschitz continuous,  i.e., there exists a fixed constant  $ L_i \geq 0$ such that 	$\|\nabla_i J_i(\mathbf{x})-\nabla_i J_i(\mathbf{y}) \|\leq L_i \|\mathbf{x}-\mathbf{y}\|$ for all $\mathbf{x}, \mathbf{y}\in \mathbb{R}^n$.
\end{assumption}


In this paper, we are interested in distributed seeking of an NE in a game $\Gamma(n, \{J_{i}\}, \{\Omega_i =\mathbb{R}\})$ where Assumption \ref{asp3},\ref{cvx},\ref{mono} and \ref{asp2} hold. Note that given  Assumption \ref{cvx} and \ref{mono}, there exists an unique NE in the game $\Gamma(n, \{J_{i}\}, \{\Omega_i=\mathbb{R}\})$ \cite{chen2022linear} and it can be   can be alternatively characterized by using the first-order optimality conditions. To be specific,  $\mathbf{x}^\star \in \mathbb{R}^n$ is an NE if and only if 
\begin{equation}\label{ne}
	 \mathbf{F(x^\star)}=0. 
\end{equation}
To cope with incomplete information in a distributed network, we assume that each agent maintains a local variable 
\begin{equation}
	\mathbf{x}_{(i)}=[\tilde x_{(i)1},\ldots,\tilde x_{(i)i-1},x_i,\tilde x_{(i)i+1},\ldots,\tilde x_{(i)n}]^\top \in \mathbb{R}^n,
\end{equation}
 which is his estimation of the joint action profile $\mathbf{x}=[x_1,x_2,
 \ldots,x_n]^\top$, where $\tilde x_{(i)j}\in\mathbb{R}$ denotes agent $i$'s estimate of $x_j$ and $\tilde x_{(i)i}=x_i\in \mathbb{R}$. The objective of this work is to design an efficient algorithm with compressed and event-triggered communication, where each agents can find an  NE  in game $\Gamma$.

\section{Event-triggered and Compressed Distributed NE Seeking Algorithm (ETC-DNES)}\label{det}
In this section, we  propose a distributed NE seeking algorithm with event-triggered and compressed communication, where compressed information is exchanged at each triggered time. We show that the proposed algorithm can converge to the optimal solution based on  carefully chosen event-triggering thresholds. Furthermore, we establish conditions under which the algorithm exhibits linear convergence.
 \subsection{Algorithm Development}
 In this subsection, we introduce the ETC-DNES algorithm which is a distributed algorithm that carefully incorporates the difference  compression technique from \cite{chen2022linear} and deterministic event-triggered communication mechanism into the conventional NE seeking algorithm from \cite{tatarenko2020geometric1}. 
 The detailed procedures of ETC-DNES are outlined in Algorithm \ref{alg1} , and we provide a summary of the notations used throughout the algorithm in Table \ref{table}.

 \begin{table}[t]
 \caption{Descriptions of notations in  Algorithm \ref{alg1}}
 \label{table}
\centering
\begin{tabular}{cc}
\hline
Symbol & Description \\ \hline
\qquad\qquad$\mathbf{x}_{(i)}$       \qquad      \qquad         &  Local estimate of the joint action profile    \qquad    \\

\qquad\qquad$\mathbf{\hat x}_{(i)}$   \qquad\qquad                        &      Estimated  version of $\mathbf{x}_{(i)}$     \\
\qquad\qquad $\mathbf{\hat x}_{(i),w}$  \qquad\qquad                         &      Weighted average of $\mathbf{\hat x}_{(i)}$     \\
\qquad\qquad$\mathbf{h}_{i}$   \qquad\qquad                        &
      Reference point of $\mathbf{x}_{(i)}$       \\

\qquad\qquad$\mathbf{h}_{i,w}$  \qquad\qquad                         &      Weighted average of $\mathbf{h}_{i}$       \\
\qquad\qquad$\mathbf{q}_{i}$   \qquad\qquad                       &   Compressed value of $
\mathbf{x}_{(i)}-\mathbf{h}_{i}$     \\ 
\qquad\qquad$\mathbf{\tilde q}_{i}$   \qquad\qquad                       &Latest sent   value      \\ 
\qquad\qquad$\mathbb{I}_{i,k}$   \qquad\qquad                       &Triggering indicator    \\ 

\hline
\end{tabular}

\end{table}

 \begin{algorithm}[t]
\caption{ETC-DNES: Event-triggered and Compressed Distributed Nash Equilibrium Seeking  Algorithm}
\label{alg1}
 {\bf Input:} %
Stopping time $K$, gradient stepsize $\eta$, consensus stepsize $\gamma$, scaling parameter $\alpha>0$, and initial values $\mathbf{x}_{(i),0}, \mathbf{h}_{i,0}, \mathbf{\tilde q}_{i,-1}$\\
{\bf Output: $\mathbf{x}_{(i),k}$}
\begin{algorithmic}[1]
\For {each agent $i\in\mathcal{V}$}
\State $\mathbf{h}_{i,w,0}=\sum\limits_{j\in \mathcal{N}_i} w_{ij}\mathbf{h}_{j,0}$
\EndFor
\For {$k=0,1,2,\ldots, K-1$} locally at each agent $i\in\mathcal{V}$
\State  
$\mathbf{q}_{i,k}=\mathcal{C}_{i,k}(\mathbf{x}_{(i),k}-\mathbf{h}_{i,k})$.
\State  Run Algorithm \ref{alg_et} and obtain the triggering indicator $\mathbb{I}_{i,k}$
\If {$\mathbb{I}_{i,k}=1$,}
 \State Send $\mathbf{q}_{i,k}$ to agent $l \in \mathcal{N}_i^+$ and set $ \mathbf{\tilde q}_{i,k}=\mathbf{q}_{i,k}$
 \Else  \State Do not send and set $ \mathbf{\tilde q}_{i,k}=\mathbf{\tilde q}_{i,k-1}$.
 \EndIf

\State $\mathbf{\hat x}_{(i),k}=\mathbf{h}_{i,k}+\mathbf{\tilde q}_{i,k}$
\State $\mathbf{h}_{i,k+1}=(1-\alpha)\mathbf{h}_{i,k}+\alpha \mathbf{\hat x}_{(i),k}$
\State Let $\mathbf{\tilde q}_{j,k}=\mathbf{ q}_{j,k}$ if receive $\mathbf{q}_{j,k}$ from agent $j \in \mathcal{N}_i^-$. Otherwise, let  $\mathbf{\tilde q}_{j,k}=\mathbf{\tilde q}_{j,k-1}$
\State $\mathbf{\hat x}_{(i),w,k} =\mathbf{h}_{i,w,k}+\sum_{j=1}^n w_{ij}\mathbf{\tilde q}_{j,k}$
\State $\mathbf{h}_{i,w,k+1}=(1-\alpha)\mathbf{h}_{i,w,k}+\alpha\mathbf{\hat x}_{(i),w,k}$
\State $\mathbf{x}_{(i),k+1}=\mathbf{x}_{(i),k}-\gamma(\mathbf{\hat x}_{(i),k}-\mathbf{\hat x}_{(i),w,k})-\gamma \eta\nabla_i J_i(\mathbf{x}_{(i),k})\mathbf{e}_i$
\EndFor
 \end{algorithmic}
\end{algorithm}

\begin{algorithm}[t]
\caption{Event-triggering Condition Calculating Algorithm}
\label{alg_et}
 {\bf Input:} %
Time iteration $k$, state value $\mathbf{q}_{i,k}$ and $\mathbf{\tilde q}_{i,k-1}$, triggering thresholds $\{\tau_{i,k}\}_{k\geq 0}$  \\
{\bf Output: $\mathbb{I}_{i,k}$}
\begin{algorithmic}[1]
\If {$k=0$,}
\State $\mathbb{I}_{i,k}=1$
\Else
\If {$||\mathbf{q}_{i,k}-\mathbf{\tilde q}_{i,k-1}||\ge\tau_{i,k}$,}
 \State $\mathbb{I}_{i,k}=1$ \Else  \State $\mathbb{I}_{i,k}=0$
 \EndIf 
\EndIf
 \end{algorithmic}
\end{algorithm}

\textbf{Transmission of Compressed Value:}  In ETC-DNES, the transmission of information occurs only when an event is triggered at iteration $k$, denoted as $\mathbb{I}_{i,k}=1$. Instead of compressing the local variable $\mathbf{x}_{(i)}$, we maintain an auxiliary variable $\mathbf{h}_{i}^k$, acting as \textit{a reference point} of $\mathbf{x}_{(i)}$, and compress the difference $\mathbf{x}_{(i)}^k-\mathbf{h}_{i}^k$ (line 5 in Algorithm \ref{alg_et}).  This choice serves to mitigate compression errors, as the error tends to vanish as the reference point $\mathbf{h}_{i,k}$ converges to $\mathbf{x}_{(i),k}$. In addition, each agent $i$ obtains \textit{a coarse representation} of $\mathbf{x}_{(i)}$, denoted as $\mathbf{\hat x}_{(i)}$, by assembling  $\mathbf{h}_{i}^k$ and the \textit{latest sent value} from its neighbors,  $\mathbf{\tilde q}_{j}, \forall j \in \mathcal{N}_i^{-}$ (line 12 in Algorithm \ref{alg_et}). Then $\mathbf{h}_{i}^{k+1}$ is obtained as the weighted average of its previous value $\mathbf{h}_{i}^k$  and $\mathbf{\hat x}_{(i)}^k$ with mixing weight $\alpha$, indicating that $\mathbf{h}_{i}^k$ is tracking the motions of $\mathbf{x}_{i}^k$ (line 13 in Algorithm \ref{alg_et}). Moreover, the variable $\mathbf{h}_{i,w}$ is a \textit{weighted averaged} version of $\mathbf{h}_{i}$, which can be regarded as a backup copy for the neighboring information (line 16 in Algorithm \ref{alg_et}). The introduction of this auxiliary variable eliminates the need to store all the neighbors' variable $\mathbf{h}_{j}$.

\textbf{Event-triggered Mechanism:} To determine when messages should be sent to neighboring agents, each agent executes Algorithm \ref{alg_et} to calculate the triggering indicator at each time iteration. At the initial iteration, $k=0$, the indicator $\mathbb{I}_{i,k}$ is set to 1, indicating that agents automatically trigger  transmissions at the start. For subsequent iterations, $k\ge 1$, each agent evaluates the triggering condition (line 4 in Algorithm \ref{alg_et}) to determine the triggering indicator's value.

Specifically, agent $i$ transmits their current compressed value, $\mathbf{q}_{i,k}$, to their neighbors and updates the latest sent value, $ \mathbf{\tilde q}_{i,k}$, to equal the transmitted value only if the difference between the current compressed value and the latest sent value, $||\mathbf{q}_{i,k}-\mathbf{\tilde q}_{i,k-1}||$, exceeds or equals a predefined event-triggering threshold $\tau_{i,k}$. Otherwise, if the difference is less than $\tau_{i,k}$, agent $i$ refrains from communication in that round and retains the previous value as the latest sent value. Smaller values of $\tau_{i,k}$ result in more frequent communication between agents. When $\tau_{i,k}=0$ for every iteration, agents always transmit their compressed value to neighbors, rendering ETC-DNES equivalent to C-DNES \cite{chen2022linear} without event-triggering laws.

 We make the following mild assumption for the event-triggering threshold in ETC-DNES. 
 \begin{assumption}\label{thres}
 	The thresholds $\{\tau_{i,k}\}_{k\ge 0}, \forall i \in \mathcal{N}$ satisfies
 	\begin{equation}
 	\tau_{i,k}\ge 0,	\qquad \lim\limits_{k\rightarrow\infty} \tau_{i,k}=0.  
 	\end{equation}
 	
 	 \end{assumption}
Assumption \ref{thres} is  a mild assumption satisfied by many sequences.  For example,  the threshold can be chosen as a fraction form $\tau_{i,k}=\frac{a_i}{(k+b_i)^2}$ with $a_i>0, b_i>0$ or an exponential form $\tau_{i,k}=C_i a_i^k$ with $C_i>0$ and $0<a_i<1$.

For the convenience  of analysis, we define the event error as $\mathbf{e}_{i,k}=\mathbf{q}_{i,k}-\mathbf{\tilde q}_{i,k}$ and 
 denote the compact form of action-profile estimates from all agents as $\mathbf{X}=[\mathbf{x}_{(1)},\mathbf{x}_{(2 )},\ldots,\mathbf{x}_{(n)}]^\top\in \mathbb{R}^{n \times n},$
 where the $i$th row is  the estimation vector $\mathbf{x}_{(i)}, i\in\mathcal{V}$. Auxiliary variables of the agents in compact form $\mathbf{H},\mathbf{H}_w, \mathbf{Q},\mathbf{\tilde Q},\mathbf{E}, \mathbf{\hat X}$ and $ \mathbf{\hat X}_w$ are defined similarly.  At $k-$th iteration, their values are denoted by $\mathbf{X}_k, \mathbf{H}_k,\mathbf{H}_{w,k}, \mathbf{Q}_k, \mathbf{\tilde Q}_k,\mathbf{E}_k,\mathbf{\hat X}_k$ and $ \mathbf{\hat X}_{w,k}$ , respectively.

 Moreover, we define a diagonal matrix 
 $$\mathbf{\tilde F}(\mathbf{X})\triangleq \text{Diag}(\nabla_1 J_1(\mathbf{x}_{(1)}),\ldots, \nabla_n J_n(\mathbf{x}_{(n)}))\in \mathbb{R}^{n \times n}.$$

 Algorithm \ref{alg1} can be written in compact form as follows:
 \begin{subequations}\label{eq:alg1}
 \begin{align}\label{ala}
&\mathbf{Q}_k=\mathcal{C}_k(\mathbf{X}_k-\mathbf{H}_k),\\ \label{alb} 
&\mathbf{\hat X}_k=\mathbf{H}_k+\mathbf{Q}_k-\mathbf{E}_k,\\ 
&\mathbf{\hat X}_{w,k}=\mathbf{H}_{w,k}+W(\mathbf{Q}_k-\mathbf{E}_k),\\\label{ald}
 &\mathbf{H}_{k+1}=(1-\alpha)\mathbf{H}_k+\alpha\mathbf{\hat X}_k,\\
  &\mathbf{H}_{w,k+1}=(1-\alpha)\mathbf{H}_{w,k}+\alpha\mathbf{\hat X}_{w,k},\\
&\mathbf{X}_{k+1}=\mathbf{X}_k-\gamma(\mathbf{\hat X}_k-\mathbf{\hat X}_{w,k})-\gamma\eta \mathbf{\tilde F}(\mathbf{X}_k),\label{alf} 
 \end{align}
\end{subequations}
 where $\mathbf{X}_{0}$ and $\mathbf{H}_{0}$ are arbitrary.  

Since $\mathbf{H}^{k}_w=W\mathbf{H}^{k}$ and $\mathbf{\hat X}^{k}_w=W\mathbf{\hat X}^{k}$  for all $k$ \cite{chen2022linear},   the state variable update in \eqref{alf} becomes
\begin{equation}\label{algx}
	\begin{aligned}
		\mathbf{X}_{k+1}&=\mathbf{X}_k-\gamma(\mathbf{\hat X}_k-W\mathbf{\hat X}_k)-\gamma\eta \mathbf{\tilde F}(\mathbf{X}_k) \\
		&=\mathbf{X}_{k}-\gamma \mathbf{F}_a(\mathbf{X}_{k})+\gamma(I-W)\mathbf{ D}_{k},\\
	\end{aligned}
\end{equation}where $\mathbf{F}_a(\mathbf{X}_k)=(I-W)\mathbf{X}_k+\eta\mathbf{\tilde F}(\mathbf{X}_k)$ denotes  the \textit{augmented mapping}  \cite{tatarenko2020geometric1} of game $\Gamma$, $\mathbf{ D}_k=\mathbf{X}_k-\mathbf{\hat X}_k$ denotes the compression error for the decision variable. 
	
\begin{lemma} \label{lml}(Lemma 1 in \cite{tatarenko2020geometric1})
	Given Assumption \ref{asp2}, the augmented mapping $\mathbf{F}_a$	of game $\Gamma$ is Lipschitz continuous  with  $L_F=\eta L_m+||I-W||_{\text{F}}$, where $L_m=\max_iL_i$. 
\end{lemma}
\begin{lemma}\label{lmmu} (Lemma 3 in \cite{tatarenko2020geometric1})
	Given Assumption \ref{asp3}, \ref{mono} and \ref{asp2}, the augmented mapping $\mathbf{F}_a$	of game $\Gamma(n, \{J_{i}\}, \{\Omega_i\})$ is restricted strongly monotone to any NE matrix $\mathbf{X}^\star=\mathbf{1}(\mathbf{x}^\star)^\top$ with the constant $\mu_F=\min\{b_1,b_2\}>0$, where $b_1=\eta\mu_r/2n, b_2=(\beta^2\tilde \lambda_{\text{min}}(I-W)/(\beta^2+1))-\eta^2 L_m$ and $\beta$ is a positive constant such that $\beta^2+2\beta=\frac{\mu_r}{2n\eta L_m}$.  
\end{lemma}
 \subsection{Convergence Analysis}
 In this subsection, we analyze the convergence performance of ETC-DNES. Let $\mathcal{F}_k$ be the $\sigma-$algebra generated by $\{\mathbf{ D}_{0},\mathbf{ D}_{1}, \ldots, \mathbf{D}_{k-1}\}$, and denote $\mathbb{E}_{\xi}[ \cdot| \mathcal{F}_k]$ as the conditional expectation  given $\mathcal{F}_k$. Now, we present the following lemmas.

\begin{lemma}\label{lme}
	Suppose Assumption \ref{thres} holds and   $0<\beta<1$. Under ETC-DNES, there is
$
	\lim_{k\rightarrow \infty} \sum\limits_{r=0}^k \beta^{k-r}||\mathbf{E}_r||^2_{\text{F}}=0.
	$
	\end{lemma}
\begin{proof}
 We first show that  $\lim_{k\rightarrow\infty} ||\mathbf{E}_k||^2_{\text{F}}=0$. 	From the event-triggering condition in Algorithm 
	\ref{alg1}, we can obtain that for all realization of $\mathbf{q}_{i,k}, \forall k\ge 0$, if $||\mathbf{q}_{i,k+1}- \mathbf{\tilde q}_{i,k}||\ge \tau_{i,k+1}$, then $||\mathbf{q}_{i,k+1}- \mathbf{\tilde q}_{i,k+1}||=0\le\tau_{i,k+1}$. Otherwise, $||\mathbf{q}_{i,k+1}- \mathbf{\tilde q}_{i,k}||< \tau_{i,k+1}$. Since $\mathbf{\tilde q}_{i,k+1}=\mathbf{\tilde q}_{i,k}$ in such case, we still have $||\mathbf{q}_{i,k+1}- \mathbf{\tilde q}_{i,k+1}||\le\tau_{i,k+1}$. Recall that $\mathbf{q}_{i,0}= \mathbf{\tilde q}_{i,0}$, we can obtain $||\mathbf{q}_{i,k}- \mathbf{\tilde q}_{i,k}||\le\tau_{i,k}, \forall k\ge 0$, i.e., $||\mathbf{e}_{i,k}||\leq \tau_{i,k}, \forall k\ge 0, i\in\mathcal{V}$. 
	
	Thus, $||\mathbf{E}_k||^2_{\text{F}}=\sum_i^n ||\mathbf{e}_{i,k}||^2\le \sum_i^n \tau_{i,k}^2$.  Given Assumption \ref{thres}, we have $\lim_{k\rightarrow\infty} \tau_{i,k}=0, \forall i\in\mathcal{N}$. Then we can obtain $\lim_{k\rightarrow\infty} ||\mathbf{E}_k||^2_{\text{F}}=0$. Based on   Lemma 3.1 in \cite{sundhar2010distributed},  $\lim_{k\rightarrow \infty} \sum\limits_{r=0}^k \beta^{k-r}||\mathbf{E}_r||^2_{\text{F}}=0$ with $0<\beta<1$.
\end{proof}
\begin{lemma}\label{lmm}
	Suppose Assumption \ref{c3} holds, the following inequality holds for ETC-DNES
	\begin{equation}
		\mathbb{E}_{\xi}[||\mathbf{X}_k-\mathbf{\hat X}_k||^2_{\text{F}}\mid \mathcal{F}_{k}]\leq 2C||\mathbf{X}_k-\mathbf{H}_k||^2_{\text{F}}+2||\mathbf{E}_k||_{\text{F}}^2.
	\end{equation}
\end{lemma}
\begin{proof}
 From Equation \eqref{asp3} and  \eqref{alb}, we have  
\begin{equation}\label{comp}
\begin{aligned}
&\mathbb{E}_{\xi}[||\mathbf{X}_k-\mathbf{\hat X}_k||^2_{\text{F}}\mid \mathcal{F}_{k}]\\&=\mathbb{E}_{\xi}[||\mathbf{X}_k-\mathbf{H}_k-\mathcal{C}_k(\mathbf{X}_k-\mathbf{H}_k)+\mathbf{E}_k||^2_{\text{F}}\mid \mathcal{F}_{k}]	\\
&\leq 2C||\mathbf{X}_k-\mathbf{H}_k||^2_{\text{F}}+2||\mathbf{E}_k||_{\text{F}}^2. \\
\end{aligned} 
\end{equation} 
\end{proof}

The following lemma is crucial for establishing a linear system of inequalities that  bound $\mathbb{E}_{\xi}[||\mathbf{X}_{k}-\mathbf{X}^\star||^2_{\text{F}}]$ and $\mathbb{E}_{\xi}[||\mathbf{X}_{k}-\mathbf{H}_{k}||^2_{\text{F}}]$.  \begin{lemma}\label{mainl}
 	Given Assumptions \ref{asp3}--\ref{asp2}, when $\alpha \in(0,\frac{1}{r}]$, the following  component-wise inequalities holds  for ETC-DNES 
\begin{equation}\label{inequlity}
	\mathbf{V}_{k+1}\le \mathbf{A}\mathbf{V}_{k}+\mathbf{B}||\mathbf{E}_k||^2_{\text{F}},
\end{equation}	
	where 	$	\mathbf{V}_{k}=\begin{bmatrix}
	\mathbb{E}_{\xi}[||\mathbf{X}_{k}-\mathbf{X}^\star||^2_{\text{F}}]&\mathbb{E}_{\xi}[||\mathbf{X}_{k}-\mathbf{H}_{k}||^2_{\text{F}}]\\
\end{bmatrix}^\top,	
$
and the  elements of transition matrix $\mathbf{A} = [a_{ij}]$ and $\mathbf{B} = [b_{ij}]$ are respectively given by
	\vspace*{-1mm}
	\begin{equation}
	\begin{aligned}
		&\mathbf{A}=\begin{bmatrix}
	c_1(1+L_F^2\gamma^2-2\mu_F\gamma)&c_2\gamma^2\\
	c_3\gamma^2 L_F^2&c_x+c_4\gamma^2\\
	\end{bmatrix},
\\
&\mathbf{B}=\begin{bmatrix}
 c_2\gamma^2&c_4\gamma^2
\end{bmatrix}^\top,	
	\end{aligned}
	\end{equation}
with $c_1=\frac{2L_F^2-\mu_F^2}{2L_F^2-2\mu_F^2},c_2=\frac{2c_1||I-W||_{\text{F}}^2C}{c_1-1},c_3=\frac{4-2\alpha r \delta}{\alpha r \delta}, c_4=2c_3C||(I-W)||^2_{\text{F}},c_x=\frac{2-\alpha r \delta}{2}$.
 \end{lemma}
 \begin{proof}
We derive two inequalities in terms of NE-seeking error and compression error, respectively. 
\textit{NE-seeking error:}    

Based on Proposition 1 in \cite{tatarenko2020geometric1} , we conclude that $\mathbf{X}^\star=\mathbf{X}^\star-\gamma\mathbf{F}_a(\mathbf{X}^\star), \forall \gamma>0$.  Thus,   we can obtain 
\begin{equation}\label{eqxstar1}
	\begin{aligned}
		&||\mathbf{X}_{k+1}-\mathbf{X}^\star||_{\text{F}}^2\leq \frac{c_1}{c_1-1}\gamma^2||I-W||_{\text{F}}^2||\mathbf{ D}_{k}||_{\text{F}}^2,\\
		&+c_1||\mathbf{X}_{k}-\gamma \mathbf{F}_a(\mathbf{X}_{k})-\mathbf{X}^\star+\gamma\mathbf{F}_a(\mathbf{X}^\star)||_{\text{F}}^2, \\
	\end{aligned}
\end{equation} 
where $c_1>1$. 

Next, we bound
\begin{equation}\label{eqxstar}
	\begin{aligned}
		&||\mathbf{X}_{k}-\gamma \mathbf{F}_a(\mathbf{X}_{k})-\mathbf{X}^\star+\gamma\mathbf{F}_a(\mathbf{X}^\star)||_{\text{F}}^2=||\mathbf{X}_{k}-\mathbf{X}^\star||_{\text{F}}^2\\
		&+\gamma^2|| \mathbf{F}_a(\mathbf{X}_{k})-\mathbf{F}_a(\mathbf{X}^\star)||_{\text{F}}^2-2\gamma\langle \mathbf{F}_a(\mathbf{X}_{k})-\mathbf{F}_a(\mathbf{X}^\star),\mathbf{X}_{k}-\mathbf{X}^\star\rangle\\
		&\le (1+\gamma^2L_F^2-2\gamma\mu_F)||\mathbf{X}_{k}-\mathbf{X}^\star||_{\text{F}}^2,
	\end{aligned}
\end{equation}
where the last inequality is based on Lemma \ref{lml} and Lemma \ref{lmmu}. 

From \eqref{eqxstar},  we have $\min\{1+\gamma^2L_F^2-2\gamma\mu_F\}=(1-\mu_F^2/L_F^2)>0$. Combining \eqref{eqxstar1} and \eqref{eqxstar}, we can have the following inequality by taking $c_1=(2L_F^2-\mu_F^2)/(2L_F^2-2\mu_F^2)$.

\begin{equation}
	\begin{aligned}
		||\mathbf{X}_{k+1}-\mathbf{X}^\star||_{\text{F}}^2&\le c_1(1+\gamma^2L_F^2-2\gamma\mu_F)||\mathbf{X}_{k}-\mathbf{X}^\star||_{\text{F}}^2 \\
		&+\frac{c_1\gamma^2||I-W||_{\text{F}}^2}{c_1-1}||\mathbf{ D}_{k}||_{\text{F}}^2.
	\end{aligned}
\end{equation}

Then, from Lemma \ref{lmm}, we can obtain
\begin{equation}\label{nee}
	\begin{aligned}
		&\mathbb{E}_{\xi}[||\mathbf{X}_{k+1}-\mathbf{X}^\star||_{\text{F}}^2\mid \mathcal{F}^k]\\
		&\le c_1(1+\gamma^2L_F^2-2\gamma\mu_F)\mathbb{E}_{\xi}[||\mathbf{X}_{k}-\mathbf{X}^\star||_{\text{F}}^2\mid \mathcal{F}^k] \\
		&+c_2\gamma^2\mathbb{E}_{\xi}[||\mathbf{X}_{k}- \mathbf{H}_{k}||_{\text{F}}^2\mid \mathcal{F}^k]+c_2\gamma^2||\mathbf{E}_{k}||_{\text{F}}^2,
	\end{aligned}
\end{equation}
where $c_2=\frac{2c_1||I-W||_{\text{F}}^2C}{c_1-1}$. 

\textit{Compression error of the decision variable:} 

Denote $\mathcal{C}_r^k(\mathbf{X}_{k})=\mathcal{C}^k(\mathbf{X}_{k})/r$,  according to \eqref{ald}, for $0<\alpha\leq \frac{1}{r}$, we have 
\begin{equation}\label{cpe1}
	\begin{split}
		&    	||\mathbf{X}_{k+1}- \mathbf{H}_{k+1}||^2\\
		=&||\mathbf{X}_{k+1}-\mathbf{X}_{k}+\alpha r(\mathbf{X}_{k}- \mathbf{H}_{k}-\mathcal{C}_r^k(\mathbf{X}_{k}- \mathbf{H}_{k}))\\
		&+(1-\alpha r)(\mathbf{X}_{k}- \mathbf{H}_{k})||^2\\
		\leq &\tau_2\Big [\alpha r||\mathbf{X}_{k}- \mathbf{H}_{k}-\mathcal{C}_r^k(\mathbf{X}_{k}- \mathbf{H}_{k})||^2\\
		&+(1-\alpha r)||\mathbf{X}_{k}- \mathbf{H}_{k}||^2\Big ]+\frac{\tau_2}{\tau_2-1}||\mathbf{X}_{k+1}-\mathbf{X}_{k}||^2,\\
	\end{split}
\end{equation}
where $\tau_2=\frac{2-\alpha r \delta}{2-2\alpha r \delta}>1$. 

Taking conditional expectation on both sides of \eqref{cpe1}, we obtain
\begin{equation}\label{cpe2}
	\begin{split}
		&\mathbb{E}_{\xi}[||\mathbf{X}_{k+1}- \mathbf{H}_{k+1}||^2\mid \mathcal{F}^k]\\
		&\leq \tau_2[\alpha r (1-\delta)+(1-\alpha r)]\mathbb{E}_{\xi}[||\mathbf{X}_{k}- \mathbf{H}_{k}||^2\mid \mathcal{F}^k]\\
		&+\frac{\tau_2}{\tau_2-1}\mathbb{E}_{\xi}[||\mathbf{X}_{k+1}-\mathbf{X}_{k}||^2\mid \mathcal{F}^k],\\ 
	\end{split}
\end{equation}
where the inequality holds based on Assumption \ref{c3}. 

Moreover, based on Lemma \ref{lmm}, we have
\begin{equation}\label{cpe3}
	\begin{split}
		&\mathbb{E}_{\xi}[||\mathbf{X}_{k+1}-\mathbf{X}_{k}||^2\mid \mathcal{F}^k]\\
		&\leq 2\gamma^2C||(I-W)||^2 (2\mathbb{E}_{\xi}[||\mathbf{X}_{k}-\mathbf{H}_{k}||^2\mid \mathcal{F}^k]+2||\mathbf{E}_{k}||_{\text{F}}^2)\\
		&+2\gamma^2 L_F^2\mathbb{E}_{\xi}[||\mathbf{X}_{k}- \mathbf{X}^{\star}||^2\mid \mathcal{F}^k] \\
	\end{split}
\end{equation}
Bringing \eqref{cpe3} into \eqref{cpe2} and denoting $c_3=\frac{2\tau_2 }{\tau_2-1}=\frac{4-2\alpha r \delta}{\alpha r \delta}, c_x=\tau_2[\alpha r (1-\delta)+(1-\alpha r)]=\tau_2(1-\alpha r\delta)=\frac{2-\alpha r \delta}{2}<1$, we have
\begin{equation}\label{cpe4}
	\begin{split}
		&\mathbb{E}_{\xi}[||\mathbf{X}_{k+1}- \mathbf{H}_{k+1}||^2_{\text{F}}\mid \mathcal{F}^k]\leq c_3\gamma^2 L_F^2\mathbb{E}_{\xi}[||\mathbf{X}_{k}-\mathbf{X}^{\star}||_{\text{F}}^2\mid \mathcal{F}^k]\\
		&+ (c_x+c_4\gamma^2)\mathbb{E}_{\xi}[||\mathbf{X}_{k}- \mathbf{H}_{k}||^2\mid \mathcal{F}^k]+c_4\gamma^2||\mathbf{E}_{k}||_{\text{F}}^2,\\
	\end{split}
\end{equation}
where $c_4=2c_3C||(I-W)||^2_{\text{F}}$. 

Combing  equalities \eqref{nee} and \eqref{cpe4} completes the proof. 
\end{proof}
\begin{remark}
	Lemma \ref{mainl} shows that the behavior of the error vector $\mathbf{V}_{k}$ relies on the transition matrix $\mathbf{A}$ as well as the additive error terms due to the event-triggered communication.  \end{remark}

The following theorem shows the convergence properties of  ETC-DNES.
\begin{theorem}\label{th1}
	Let Assumptions \ref{asp3}--\ref{thres} hold. When $\gamma=\mu_F/L_F^2$,  the scaling parameter $\alpha$ and the gradient stepsize $\eta$ satisfy 

	\begin{equation*}\label{pa2}
	\alpha \in\left(0,\frac{1}{r}\right],\quad 	\eta\le \min\Big\{\frac{2n}{\mu_r}\sqrt{\frac{1-c_x}{m_1}},\sqrt{\frac{\tilde \lambda_{\text{min}}(I-W)}{2L_m}}, \frac{\mu_r}{6nL_m}\Big\},
	\end{equation*}
	where $m_1=\frac{4c_2c_3}{||I-W||_{\text{F}}^4}+\frac{1}{4||I-W||_{\text{F}}^2}+\frac{c_4}{||I-W||_{\text{F}}^4}$, and  the spectral radius of $\mathbf{A}$ is less than $1$. Thus, the NE-seeking error $\mathbb{E}_{\xi}[||\mathbf{ X}^{k+1}-\mathbf{X}^\star||^2_{\text{F}}]$ and the compression error $\mathbb{E}_{\xi}[||\mathbf{X}^{k+1}-\mathbf{H}^{k+1}||^2_{\text{F}}]$ both
 converge to $0$ under ETC-DNES. \end{theorem} 
 \begin{proof}
 	Given  $\gamma=\mu_F/L_F^2$,  we have $a_{11}=c_1(1+L_F^2\gamma^2-2\mu_F\gamma)=(1-\mu_F^2/2L_F^2)<1, a_{12}=\frac{\mu_F^2 c_2}{L_F^4}, a_{21}=\frac{c_3\mu_F^2}{L_F^2}$ and $a_{22}=c_x+c_4\frac{\mu_F^2}{L_F^4}$.
 	
 	Since $0<\alpha\le \frac{1}{r}$, Lemma \ref{mainl} holds for ETC-DNES. To establish the convergence of ETC-DNES, it is essential to determine the range of the stepsize that ensures $\rho(\mathbf{A})<1$. Based on Corollary 8.1.29 in \cite{horn2012matrix}, it can be seen  that $\rho(\mathbf{A})<1$ if there exists a positive vector $\mathbf{\epsilon}:=[\epsilon_1,\epsilon_2]^\top \in \mathbb{R}_{++}^2$ for which the inequality
 	\begin{equation}\label{pr1}
 		\mathbf{A\epsilon}\le (1-\frac{\mu_F^2}{4L_F^2})\mathbf{\epsilon}
 	\end{equation}
 	is satisfied. 
 	
 	Since 	$\mathbf{A\epsilon}=[a_{11}\epsilon_1+a_{12}\epsilon_2, a_{21}\epsilon_1+a_{22}\epsilon_2]^\top$, inequality \eqref{pr1} is equivalent to 
 	\begin{subequations}
 		\begin{align}\label{ine:rho1}
 			&(1-\frac{\mu_F^2}{2L_F^2})\epsilon_1+\frac{\mu_F^2 c_2}{L_F^4}\epsilon_2 \leq (1-\frac{\mu_F^2}{4L_F^2})\epsilon_1,\\\label{ine:rho2}
 			&\frac{c_3\mu_F^2}{L_F^2}\epsilon_1+(c_x+c_4\frac{\mu_F^2}{L_F^4})\epsilon_2\leq (1-\frac{\mu_F^2}{4L_F^2})\epsilon_2.
 		\end{align}
 	\end{subequations}
 	
 	Next, we derive the First inequalities  \eqref{ine:rho1} and  \eqref{ine:rho2}. 
 	
 	First, it can be easily seen that inequality \eqref{ine:rho1} holds if $$\frac{4c_2}{L_F^2}\epsilon_2=\epsilon_1.$$

 	Moreover, given $\frac{4c_2}{L_F^2}\epsilon_2=\epsilon_1$, inequality \eqref{ine:rho2}  is equivalent to
 	\begin{equation}
 		\frac{4c_2c_3\mu_F^2}{L_F^4}\epsilon_2+(\frac{\mu_F^2}{4L_F^2}+c_4\frac{\mu_F^2}{L_F^4})\epsilon_2\leq (1-c_x)\epsilon_2.
 	\end{equation}
 	
 	Note that $L_F^2>||I-W||_{\text{F}}^2$ and denote 
 	\begin{equation}\label{m1}
 		m_1=\frac{4c_2c_3}{||I-W||_{\text{F}}^4}+\frac{1}{4||I-W||_{\text{F}}^2}+\frac{c_4}{||I-W||_{\text{F}}^4}.
 	\end{equation} 
 	
 	Inequality \eqref{ine:rho2} can be verified when
 	\begin{equation}\label{muf}
 		\mu_F\le \sqrt{\frac{1-c_x}{m_1}}.
 	\end{equation}

 	Since $\mu_F=\min\{b_1,b_2\}>0$, the inequality \eqref{muf} is guaranteed by
 	\begin{equation}\label{ine:b}
 		\begin{aligned}
 			&b_1=\eta\mu_r/2n\le \sqrt{\frac{1-c_x}{m_1}}, \\
 			&b_2=(\beta^2\tilde \lambda_{\text{min}}(I-W)/(\beta^2+1))-\eta^2 L_m>0.
 		\end{aligned}
 	\end{equation}
 	
 	Based on 	Lemma \ref{lmmu}, we have $\beta^2+2\beta=
 	\frac{\mu_r}{2n\eta L_m}$. Then, it can be derived that $\beta>1$ when $\frac{\mu_r}{2n\eta L_m}>3$, i.e., $\beta>1$ when $\eta<\frac{\mu_r}{6nL_m}$.
 	
 	Hence, the sufficient condition for inequality \eqref{ine:b} is
 	\begin{equation}\label{m2}
 		\eta\le \min\Big\{\frac{2n}{\mu_r}\sqrt{\frac{1-c_x}{m_1}},\sqrt{\frac{\tilde \lambda_{\text{min}}(I-W)}{2L_m}}, \frac{\mu_r}{6nL_m}\Big\}.
 	\end{equation}

 	To wrap up, if the positive constants $\epsilon_1, \epsilon_2$ and the stepsize $\eta$ satisfy the following conditions, 
 	\begin{equation}
 		\begin{split}
 			&\epsilon_1=\frac{4c_2}{L_F^2}\epsilon_2,\quad \epsilon_2>0,\\
 			&	\eta\le \min\Big\{\frac{2n}{\mu_r}\sqrt{\frac{1-c_x}{m_1}},\sqrt{\frac{\tilde \lambda_{\text{min}}(I-W)}{2L_m}}, \frac{\mu_r}{6nL_m}\Big\},
 		\end{split}
 	\end{equation}
 	the linear system of  element-wise inequalities in \eqref{pr1} can be shown.

According to the above results, i.e., $\rho(\mathbf{A})<1$, we continue to establish the remaining proof of convergence of ETC-DNES. We recursively iterate \eqref{inequlity} as follows:
\begin{equation}\label{v1}
	\mathbf{V}_k \le \mathbf{A}^k \mathbf{V}_0+\sum\limits_{r=0}^{k-1}  \mathbf{A}^{k-r-1}\mathbf{B}||\mathbf{E}_r||^2_{\text{F}}. 
\end{equation}
Since $ \mathbf{A}^2$ is a positive matrix, $ \mathbf{A}$ is primitive (Theorem 8.5.2 in \cite{horn2012matrix}). Then each element of $ \mathbf{A}^k$ decreases at the rate of $\rho( \mathbf{A})$ (Theorem 8.5.1 in \cite{horn2012matrix}), i.e.,
\begin{equation}\label{v2}
	\mathbf{A}^k\le C_A \rho^k_ {A}\mathbf{1}\mathbf{1}^\top,
\end{equation} 
where $C_A$ is a positive constant.

Hence, we have
\begin{equation}\label{v3}
	\begin{aligned}
		&	\sum\limits_{r=0}^{k-1}  \mathbf{A}^{k-r-1}\mathbf{B}||\mathbf{E}_r||^2_{\text{F}}\leq \sum\limits_{r=0}^{k-1}  C_A \rho^{k-r-1}_{A}\mathbf{1}\mathbf{1}^\top \mathbf{B}||\mathbf{E}_r||^2_{\text{F}}\\
		&=(c_2+c_4)\gamma^2 C_A\sum\limits_{r=0}^{k-1}\rho^{k-r-1}_{A}||\mathbf{E}_r||^2_{\text{F}}\mathbf{1}. 
	\end{aligned}
\end{equation}
Substituting \eqref{v2} and \eqref{v3} in \eqref{v1} gives
\begin{equation}\label{v4}
	\mathbf{V}_k\le C_A\big( \rho^k_ {A}\mathbf{1}^\top\mathbf{V}_0+(c_2+c_4)\gamma^2\sum\limits_{r=0}^{k-1}\rho^{k-r-1}_{A}||\mathbf{E}_r||^2_{\text{F}}\big)\mathbf{1}. 
\end{equation}

With Lemma \ref{lme}, we have 
$
\lim\limits_{k\rightarrow \infty}	\sum\limits_{r=0}^{k-1}\rho^{k-r-1}_{A}||\mathbf{E}_r||^2_{\text{F}}=0.
$
Thus, we have $\lim\limits_{k\rightarrow \infty} \mathbf{V}_k=0$, which completes the proof. 
 \end{proof}

\begin{corollary}\label{co1}
Suppose Assumptions \ref{asp3}--\ref{thres} hold and the algorithm parameters satisfy \eqref{pa2}.
	If the threshold sequence is chosen in an exponential form, i.e., $\tau_{i,k}=C_i p_i^k$ with $C_i>0$ and $0<p_i<1$, the error vector $\mathbf{V}_k$ converges to $0$ linearly under Algorithm \ref{alg1}. 
	\end{corollary}
\begin{proof}
Since $\tau_{i,k}=C_i p_i^k$, we can obtain $||\mathbf{E}_k||^2_{\text{F}}\le \sum_i^n \tau_{i,k}^2\le C\bar p^{2k}$, where $C=\max_i C_i$, $\bar p=\max_i p_i$.  Then equation \eqref{v4} can be written as
\begin{equation*}
	\begin{aligned}
		\mathbf{V}_k&\le C_A\big( \rho^k_ {A}\mathbf{1}^\top\mathbf{V}_0+(c_2+c_4)\gamma^2 C\sum\limits_{r=0}^{k-1}\rho^{k-r-1}_{A}\bar p^{2r}\big)\mathbf{1}\\
		&\le C_A\big( \rho^k_ {A}\mathbf{1}^\top\mathbf{V}_0+(c_2+c_4)\gamma^2 C\sum\limits_{r=0}^{k-1}\rho^{k-r-1}_{A}\bar p^{r}\big)\mathbf{1}\\
		&\le  C_A\big( \mathbf{1}^\top\mathbf{V}_0m^k+(c_2+c_4)\gamma^2 Cm^{k-1}\big)\mathbf{1},
	\end{aligned}
\end{equation*}
where  $0<\max\{\rho_ {A}, \bar p\}<m<1$ and the last inequality can be obtained from Lemma 2 in \cite{hayashi2018event}.

Hence, the vector $\mathbf{V}_k$ decreases at the linear rate $\mathcal{O}(m^k)$, where $0<m<1$, which implies that the error vector converges to $0$ linearly.  
\end{proof}

\begin{remark}
Theorem \ref{th1} demonstrates the asymptotic convergence of each agent's state to the  NE. In the compression-based method \cite{chen2022linear} without event-triggered mechanism, the convergence rate is solely dependent on the spectral radius $\rho(\mathbf{A})$. However, in our proposed algorithm that integrates both event-triggering and compression, the upper bound of the error vector $\mathbf{V}(k)$ in \eqref{inequlity} indicates that the convergence rate depends on both the threshold of the triggering condition and the spectral radius $\rho(\mathbf{A})$. In general, the convergence rate of the event-triggered algorithm is slower compared to its non-event-triggered counterpart due to reduced information exchange. Nevertheless, Corollary \ref{co1} reveals that with a suitably chosen triggering condition, the event-triggered algorithm can achieve linear convergence to the optimal solution. Therefore, by judiciously selecting the trigger condition's threshold, the proposed event-triggered method effectively reduces communication load without significantly impacting the convergence rate in terms of its order.
\end{remark}

\begin{remark}
	All the above results can be adopted for games with different dimensions of the action sets.
\end{remark}

\section{Stochastic Event-triggered and Compressed Distributed NE Seeking Algorithm  (SETC-DNES) }
In this section, we introduce a  stochastic event-triggered mechanism for  distributed NE seeking. This mechanism incorporates a random variable into the event-triggering condition, resulting in different probabilities of triggering events based on local information. We demonstrate that the proposed algorithm (SETC-DNES)  achieves linear convergence to the NE.
\begin{algorithm}[t]
\caption{Stochastic Event-triggering Condition Calculating Algorithm}
\label{alg_set}
 {\bf Input:} %
Time iteration $k$, state value $\mathbf{q}_{i,k}, \mathbf{\tilde q}_{i,k-1}, \mathbf{x}_{(i),k}, \mathbf{h}_{i,k}$,  event-trigger parameter $\kappa>1$, $0<a<1$  \\
{\bf Output: $\mathbb{I}_{i,k}$}
\begin{algorithmic}[1]
\If {$k=0$,}
\State $\mathbb{I}_{i,k}=1$
\Else
\State Extract $\zeta_{i,k} \in(a,1)$ from an  arbitrary stationary ergodic \hspace*{10mm}random process

\If {$\zeta_{i,k}>\kappa \exp(\frac{-||\mathbf{q}_{i,k}-\mathbf{\tilde q}_{i,k-1}||}{||\mathbf{x}_{(i),k}-\mathbf{h}_{i,k}||})$},
 \State $\mathbb{I}_{i,k}=1$ \Else  \State $\mathbb{I}_{i,k}=0$
 \EndIf 
\EndIf
 \end{algorithmic}
\end{algorithm}

\subsection{Algorithm Development}
In this subsection, we extend the deterministic event-triggered mechanism from Section \ref{det} to a stochastic one. Given that SETC-DNES and ETC-DNES only differ  in the triggering condition, we will not provide a detailed overview of SETC-DNES. Instead, we present Algorithm \ref{alg_set} to illustrate the stochastic event-triggered mechanism.

\textbf{Stochastic Event-triggered Mechanism:}  In Algorithm \ref{alg_set}, a novel stochastic event-triggered mechanism is proposed.  The intuition behind the proposed stochastic event-triggered method is to extend the deterministic one in Section \ref{det} by assigning a probability for each event, rather than triggering whenever the defined event occurs. 

Here is how it works: at each iteration $k\geq 1$, a random variable $\zeta_{i,k}$ is drawn from a stationary ergodic random process with an identical probability density function $f_{i,\zeta}$ for all agents $i\in \mathcal{V}$. Then, each agent $i\in \mathcal{V}$ evaluates the triggering condition (line 5 in Algorithm \ref{alg_set}). If the condition is met, the triggering indicator is set to $1$, and agent $i$ communicates the compressed value to all its neighbors. Otherwise, no communication takes place at that iteration.

Differing from the event-triggering condition in Algorithm \ref{alg_et}, which requires a threshold $\tau_{i,k}$, this approach dispenses with the need for $\tau_{i,k}$. Instead, agent $i$ calculates the ratio between $||\mathbf{q}_{i,k}-\mathbf{\tilde q}_{i,k-1}||$ and $||\mathbf{x}_{(i),k}-\mathbf{h}_{i,k}||$. Then, it compares $\zeta_{i,k}$ with $\kappa \exp(-||\mathbf{q}_{i,k}-\mathbf{\tilde q}_{i,k-1}||/||\mathbf{x}_{(i),k}-\mathbf{h}_{i,k}||)$ to evaluate the triggering condition. Crucially, as $||\mathbf{x}_{(i),k}-\mathbf{h}_{i,k}||$ diminishes over time (refer to Section \ref{analysis2} for further details), it's ensured that the probability of communication increases monotonically with $||\mathbf{q}_{i,k}-\mathbf{\tilde q}_{i,k-1}||$. This means that higher probabilities of communication are assigned to more urgent cases corresponding to larger values of $||\mathbf{q}_{i,k}-\mathbf{\tilde q}_{i,k-1}||$.  To provide an illustrative example, consider a scenario where $a=0.5$, and $\zeta_{i,k}$ follows a uniformly distributed random process. In this case, the probability of event occurrence, denoted as $P[\mathbb{I}_{i,k}=1]$, can be expressed as $0.5(1-\kappa \exp(-||\mathbf{q}_{i,k}-\mathbf{\tilde q}_{i,k-1}||/||\mathbf{x}_{(i),k}-\mathbf{h}_{i,k}||))$, i.e., $P[\mathbb{I}_{i,k}=1]$ increases with the value of $||\mathbf{q}_{i,k}-\mathbf{\tilde q}_{i,k-1}||$. In instances where $\zeta_{i,k}$ remains a positive constant, the stochastic scheme  transforms into a deterministic one.

%


Note that except for the calculation of the triggering indicator, the state update rule in SETC-DNES is identical to that in ETC-DNES. Consequently, the compact form of SETC-DNES remains the same as ETC-DNES, as shown in Equations \eqref{eq:alg1}.

\subsection{Convergence Analysis}\label{analysis2}
In this subsection, we analyze the convergence performance of SETC-DNES. We first have the following inequality  when $\mathbb{I}_{i,k}=0$ in SETC-DNES:
\begin{equation}\label{eq_e}
	||\mathbf{q}_{i,k}-\mathbf{\tilde q}_{i,k-1}||\le (||\mathbf{x}_{(i),k}-\mathbf{h}_{i,k}||) (\ln \kappa-\ln \zeta_{i,k}).
\end{equation}

Given that $a<\zeta_{i,k}<1$, we can establish an upper bound for the event error in SETC-DNES as follows:
\begin{equation}\label{eq_e1}
||\mathbf{E}_k||^2_{\text{F}}=\sum_i^n	||\mathbf{e}_{i,k}||^2\leq  \sum_i^n l^2(||\mathbf{x}_{(i),k}-\mathbf{h}_{i,k}||^2)=l^2||\mathbf{X}_k-\mathbf{H}_k||^2_{\text{F}},
\end{equation}
where $l=\ln \kappa-\ln a$.

We  now proceed to derive the following lemma, which will enable us to establish a linear system of inequalities for bounding NE-seeking error and compression error.
\begin{lemma}\label{mainl2}
 	Given Assumption \ref{asp3}--\ref{asp2}, when $\alpha \in(0,\frac{1}{r}]$, the following linear system of  component-wise inequalities holds  for SETC-DNES, 
\begin{equation}\label{inequlity1}
	\mathbf{V}_{k+1}\le \mathbf{C}\mathbf{V}_{k},
\end{equation}	
where the  elements of transition matrix $\mathbf{C} = [c_{ij}]$  are  given by
	
	\begin{equation}
		\mathbf{C}=\begin{bmatrix}
	c_1(1+L_F^2\gamma^2-2\mu_F\gamma)&c_5\gamma^2\\
	c_3\gamma^2 L_F^2&c_x+c_6\gamma^2	\end{bmatrix}
	\end{equation}
	with $c_5=c_2(l^2+1)$ and  $c_6=c_4(l^2+1)$.
 \end{lemma}
 
\begin{proof}
Since  SETC-DNES and  ETC-DNES differ only in their event-triggering methods, Lemma \ref{mainl} also applies to SETC-DNES.

Bringing inequality \eqref{eq_e1} into  inequality \eqref{inequlity}, we can obtain 
\begin{equation}
\begin{aligned}
	&\mathbf{V}_{k+1}\le  \mathbf{A}\mathbf{V}_{k}+l^2\mathbf{B}||\mathbf{X}_k-\mathbf{H}_k||^2_{\text{F}}\\
	&=\begin{bmatrix}
	c_1(1+L_F^2\gamma^2-2\mu_F\gamma)&c_2\gamma^2\\
	c_3\gamma^2 L_F^2&c_x+c_4\gamma^2\\
	\end{bmatrix}\mathbf{V}_{k}+
	\begin{bmatrix}
	0&l^2c_2\gamma^2\\
	0&l^2c_4\gamma^2\\
	\end{bmatrix}\mathbf{V}_{k}\\
	&=\begin{bmatrix}
	c_1(1+L_F^2\gamma^2-2\mu_F\gamma)&c_2\gamma^2+l^2 c_2\gamma^2\\
	c_3\gamma^2 L_F^2&c_x+c_4\gamma^2+l^2 c_4\gamma^2
	\end{bmatrix}\mathbf{V}_{k}\\
	&=\mathbf{C}\mathbf{V}_{k}.\\
\end{aligned}
\end{equation}	
\end{proof}

The following theorem shows the convergence properties of SETC-DNES.
\begin{theorem}\label{th2}
	Let Assumptions \ref{asp3}--\ref{asp2} hold. When  $\gamma=\mu_L/L_F^2$, the scaling parameter $\alpha$ and the gradient stepsize $\eta$ satisfy
		\begin{equation*}
		\alpha \in\left(0,\frac{1}{r}\right],\quad 	\eta\le \min\Big\{\frac{2n}{\mu_r}\sqrt{\frac{1-c_x}{m_2}},\sqrt{\frac{\tilde \lambda_{\text{min}}(I-W)}{2L_m}}, \frac{\mu_r}{6nL_m}\Big\},
	\end{equation*} where $	m_2=\frac{4c_3c_5}{||I-W||_{\text{F}}^4}+\frac{1}{4||I-W||_{\text{F}}^2}+\frac{c_6}{||I-W||_{\text{F}}^4}$, and the spectral radius of $\mathbf{C}$ is less than $1$.  Thus, the NE-seeking error $\mathbb{E}_{\xi}[||\mathbf{ X}^{k+1}-\mathbf{X}^\star||^2_{\text{F}}]$ and the compression error $\mathbb{E}_{\xi}[||\mathbf{X}^{k+1}-\mathbf{H}^{k+1}||^2_{\text{F}}]$ of SETC-DNES  both
 converge to $0$ at the linear rate $\mathcal{O}(\rho(\mathbf{C})^k)$.  \end{theorem} 
 \begin{proof}
 	Given  $\gamma=\mu_F/L_F^2$,  we have $c_{11}=c_1(1+L_F^2\gamma^2-2\mu_F\gamma)=(1-\mu_F^2/2L_F^2)<1, c_{12}=\frac{\mu_F^2 c_5}{L_F^4}, c_{21}=\frac{c_3\mu_F^2}{L_F^2}$ and $c_{22}=c_x+c_6\frac{\mu_F^2}{L_F^4}$.
 	
 {\color{black} Since $0<\alpha\leq \frac{1}{r}$, Lemma \ref{mainl2} holds for SETC-DNES.}  	To establish the convergence of SETC-DNES, it is essential to determine the range of the stepsize that ensures $\rho(\mathbf{C})<1$. Based on Corollary 8.1.29 in \cite{horn2012matrix}, it can be seen  that $\rho(\mathbf{C})<1$ if there exists a positive vector $\mathbf{\epsilon'}:=[\epsilon_3,\epsilon_4]^\top \in \mathbb{R}_{++}^2$ for which the inequality
 	\begin{equation}\label{pr2}
 		\mathbf{C\epsilon'}\le (1-\frac{\mu_F^2}{4L_F^2})\mathbf{\epsilon'}
 	\end{equation}
 	is satisfied. 
 	
 	Since 	$\mathbf{C\epsilon'}=[c_{11}\epsilon_3+c_{12}\epsilon_4, c_{21}\epsilon_3+c_{22}\epsilon_4]^\top$, inequality \eqref{pr2} is equivalent to 
 	\begin{subequations}
 		\begin{align}\label{ine:2rho1}
 			&(1-\frac{\mu_F^2}{2L_F^2})\epsilon_3+\frac{\mu_F^2 c_5}{L_F^4}\epsilon_4 \leq (1-\frac{\mu_F^2}{4L_F^2})\epsilon_3,\\\label{ine:2rho2}
 			&\frac{c_3\mu_F^2}{L_F^2}\epsilon_3+(c_x+c_6\frac{\mu_F^2}{L_F^4})\epsilon_4\leq (1-\frac{\mu_F^2}{4L_F^2})\epsilon_4.
 		\end{align}
 	\end{subequations}
 	
 	Next, we derive the conditions that lead to inequalities  \eqref{ine:2rho1} and  \eqref{ine:2rho2}. 
 	
 	First, it can be easily seen that inequality \eqref{ine:2rho1} holds if $$\frac{4c_5}{L_F^2}\epsilon_4=\epsilon_3.$$

 	Moreover, given $\frac{4c_5}{L_F^2}\epsilon_4=\epsilon_3$, inequality \eqref{ine:2rho2}  is equivalent to
 	\begin{equation}\label{ine:mu2}
 		\frac{4c_3c_5\mu_F^2}{L_F^4}\epsilon_3+(\frac{\mu_F^2}{4L_F^2}+c_6\frac{\mu_F^2}{L_F^4})\epsilon_4\leq (1-c_x)\epsilon_4.
 	\end{equation}
 	
 	Note that $L_F^2>||I-W||_{\text{F}}^2$ and denote 
 	\begin{equation}\label{m1}
 		m_2=\frac{4c_3c_5}{||I-W||_{\text{F}}^4}+\frac{1}{4||I-W||_{\text{F}}^2}+\frac{c_6}{||I-W||_{\text{F}}^4}.
 	\end{equation} 
 	
 	Inequality \eqref{ine:mu2} can be verified when
 	\begin{equation}\label{muf2}
 		\mu_F\le \sqrt{\frac{1-c_x}{m_2}}.
 	\end{equation}

 	Since $\mu_F=\min\{b_1,b_2\}>0$, the inequality \eqref{muf2} is guaranteed by
 	\begin{equation}\label{ine:2b}
 		\begin{aligned}
 			&b_1=\eta\mu_r/2n\le \sqrt{\frac{1-c_x}{m_2}}, \\
 			&b_2=(\beta^2\tilde \lambda_{\text{min}}(I-W)/(\beta^2+1))-\eta^2 L_m>0.
 		\end{aligned}
 	\end{equation}
 	
 	Based on 	Lemma \ref{lmmu}, we have $\beta^2+2\beta=
 	\frac{\mu_r}{2n\eta L_m}$. Then, it can be derived that $\beta>1$ when $\frac{\mu_r}{2n\eta L_m}>3$, i.e., $\beta>1$ when $\eta<\frac{\mu_r}{6nL_m}$.
 	
 	Hence, the sufficient condition for inequality \eqref{ine:2b} is
 	\begin{equation}\label{m2}
 		\eta\le \min\Big\{\frac{2n}{\mu_r}\sqrt{\frac{1-c_x}{m_1}},\sqrt{\frac{\tilde \lambda_{\text{min}}(I-W)}{2L_m}}, \frac{\mu_r}{6nL_m}\Big\}.
 	\end{equation}

 	To wrap up, if the positive constants $\epsilon_3, \epsilon_4$ and the stepsize $\eta$ satisfy the following conditions, 
 	\begin{equation}
 		\begin{split}
 			&\epsilon_3=\frac{4c_5}{L_F^2}\epsilon_4,\quad \epsilon_4>0,\\
 			&	\eta\le \min\Big\{\frac{2n}{\mu_r}\sqrt{\frac{1-c_x}{m_2}},\sqrt{\frac{\tilde \lambda_{\text{min}}(I-W)}{2L_m}}, \frac{\mu_r}{6nL_m}\Big\},
 		\end{split}
 	\end{equation}
 	the linear system of  element-wise inequalities in \eqref{pr2} can be shown and we can conclude that the NE-seeking error $\mathbb{E}_{\xi}[||\mathbf{ X}^{k+1}-\mathbf{X}^\star||^2_{\text{F}}]$ and the compression error $\mathbb{E}_{\xi}[||\mathbf{X}^{k+1}-\mathbf{H}^{k+1}||^2_{\text{F}}]$  both
		 converge to $0$ at the linear rate $\mathcal{O}(\rho(\mathbf{C})^k)$, where $\rho(\mathbf{C})<1$.
		
		\end{proof}

\begin{remark}
Theorem \ref{th2} establishes the linear convergence of the proposed SETC-DNES algorithm. As indicated in \eqref{eq_e1}, the event error $\mathbf{E}_k$ is bounded by the differences between $\mathbf{X}_k$ and $\mathbf{H}_k$. Consequently, the convergence rate is independent of any specific assumptions regarding the event-triggering condition and relies solely on the spectral radius $\rho(\mathbf{C})$. Therefore, Theorem \ref{th2} effectively demonstrates the well-posedness of the proposed SETC-DNES algorithm.
\end{remark}

\section{SIMULATIONS}
In this section, we provide numerical results to validate our theoretical results and compare our proposed communication-efficient algorithms with existing methods.

\textbf{Network.} We consider a randomly generated directed communication network with $n=50$ agents. The weight matrix $W$ is

\begin{equation}
    [W]_{ij}=
   \begin{cases}
   1/(\max_{i \in V} \mathcal{|N|}_i),&\mbox{if $j \in \mathcal{N}_i^{-}$;}\\
   1-\sum_{j  \in \mathcal{N}_i^{-}}w_{ij},&\mbox{if $j=i$;}\\
   0,&\mbox{Otherwise.}
   \end{cases}
  \end{equation}

\textbf{Task.} The connectivity control game defined in \cite{chen2022linear} is considered, where the sensors in the network try to find a tradeoff between the local objective, e.g., source seeking and positioning and the global objective, e.g., maintaining connectivity with other sensors. The cost function of sensor $i$ is defined as $J_i(\mathbf{x})=l_i^c(x_i)+l_i^g(\mathbf{x})$,
	where $l_i^c(x_i)=x_i^\top r_{ii}x_i+x_i^\top r_i+b_i$, $l_i^g(\mathbf{x})=\sum\limits_{j\in S_i} c_{ij}||x_i-x_j||^2$,  $S_i\subseteq V$ and $x_i=[x_{i1},x_{i2}]^\top\in\mathbb{R}^2$ denotes the position of sensor $i$, and $r_{ii},r_i, b_i, c_{ij}>0$ are constants. The parameters are set as $r_{ii}=\begin{bmatrix}
	i&0\\0&i
\end{bmatrix}, r_i=\begin{bmatrix}i&i \end{bmatrix}^\top, b_i=i$ and $c_{ij}=1,    \forall i,j \in V$. Furthermore, there is $S_i=\{i+1\}$ for $i\in\{1,2,\ldots,49\}$ and $S_{50}=\{1\}$.
The  unique Nash equilibrium of the game is $x^\star_{ij}=-0.5$ for $i\in\{1,2,\ldots,50\}, j\in\{1,2\}$. Meanwhile,  $\mathbf{x}_{(i),0}$ are randomly generated in $[0,1]^{100}$, $\mathbf{h}_{i,0}=\mathbf{0}$.    


\textbf{Compressor.} For information compression, we implement the stochastic quantization compressor  with $b=2$ and $q=\infty$ as  shown in  \cite{liu2021linear} for all agents at each iteration. As pointed out in \cite{zhang2023innovation}, transmitting this compressed message needs $(b+1)d+l$ bits if a scalar can be transmitted with $l$ bits while maintaining sufficient precision. In contrast, each agent sends $dl$ bits at each triggered  iteration when implementing  NE seeking algorithm without compression. In this simulation, we choose $l=32$.

\textbf{Algorithms.} We compare our proposed algorithms, ETC-DNES with various triggering thresholds and SETC-DNES, against several benchmarks: C-DNES \cite{chen2022linear}, the current state-of-the-art in compressed NE seeking; deterministic event-triggered NE seeking without compression (ET-NE) \cite{xu2021event}; {\color{black} and stochastic event-triggered NE seeking without compression (SET-NE) \cite{huo2023distributed}}. All hyper-parameters used in the experiments are manually tuned and detailed in Table \ref{table1}.
 \begin{table}[htp]
 \caption{Parameter settings for different algorithms}
 \label{table1}
\centering
\begin{tabularx}{0.49\textwidth}{ll}
\hline
Algorithm & Value \\ \hline
C-DNES  \cite{chen2022linear}      &  $\eta=0.01, \gamma=0.5 , \alpha=0.05$    \\
ET-NE \cite{xu2021event}       &  $\eta=0.01, \tau_k=10/k^{1.1}$    \\
SET-NE \cite{huo2023distributed}       &  $\eta=0.01, \zeta_{i,k} \sim U(0.05,1),\kappa=1.075$   \\
ETC-DNES-1       &  $\eta=0.01, \gamma=0.5 , \alpha=0.05,  \tau_k=10/k^{1.1}$  \\
ETC-DNES-2       &  $\eta=0.01, \gamma=0.5 , \alpha=0.05, \tau_k=50\times 0.99^k$  \\
SETC-DNES        &  $\eta=0.01, \gamma=0.5, \alpha=0.05,  $  \\
&$\zeta_{i,k} \sim U(0.5,1),\kappa=1.5$\\

\hline
\end{tabularx}
\end{table}

%
%

%
%

\begin{figure}[t]
	\begin{minipage}[t]{0.25\textwidth}
		\centering
		\includegraphics[width=\textwidth]{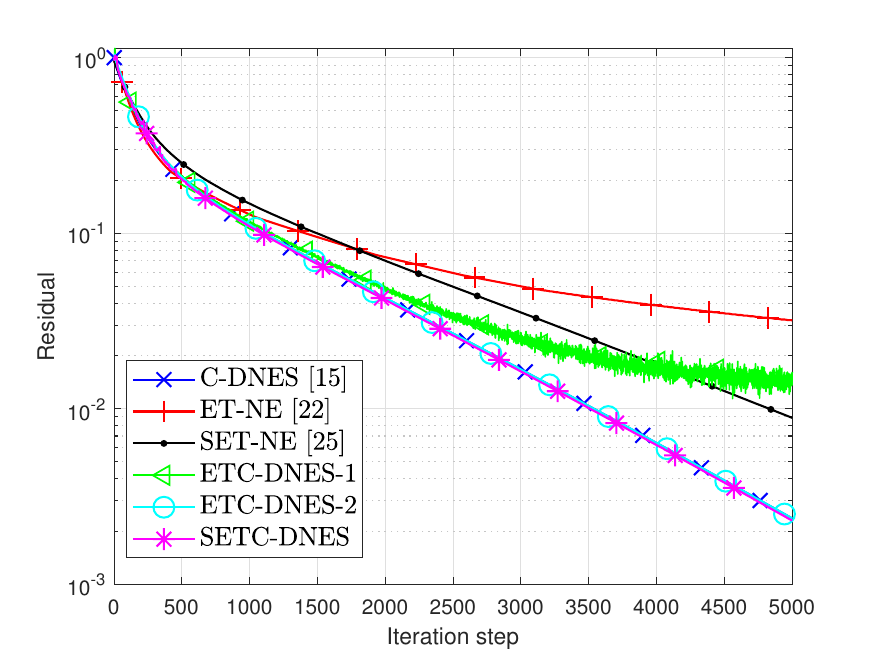}
		\centerline{(a)}
	\end{minipage}
	\begin{minipage}[t]{0.25\textwidth}
		\centering
		\includegraphics[width=\textwidth]{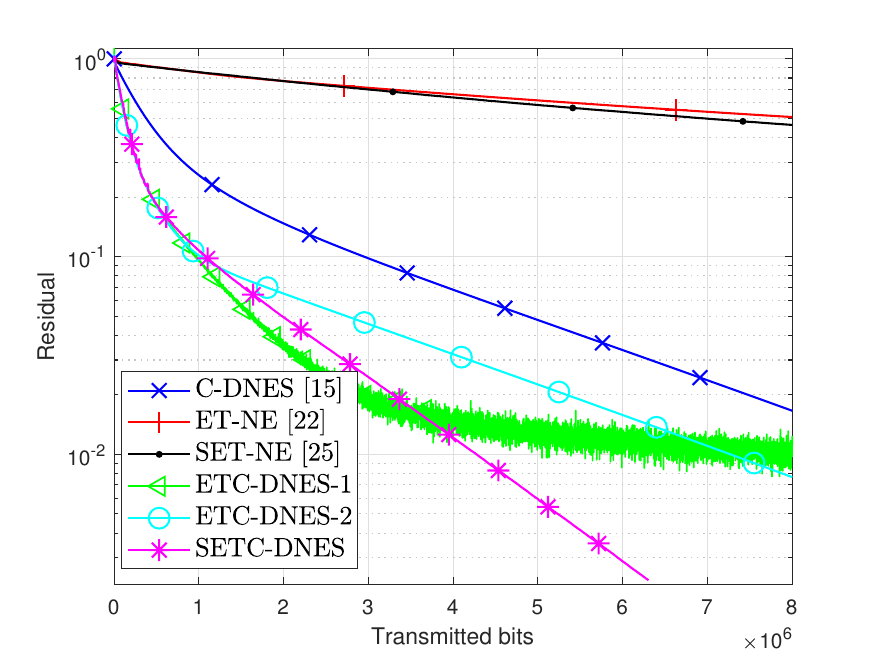}
		\centerline{(b)}
	\end{minipage}
	\caption{Evolution of Residual with respect to  (a) the number of
		iterations and (b) the number of
		transmitted bits.}	\label{simu}
\end{figure}

\begin{figure}[t]
	\begin{minipage}[htp]{0.25\textwidth}
		\centering
		\includegraphics[width=\textwidth]{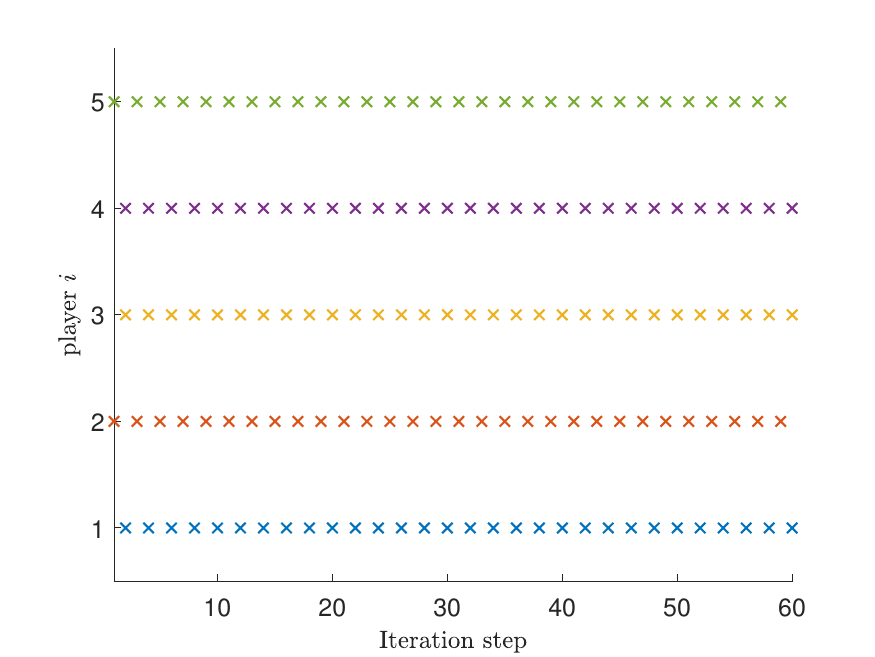}
		\centerline{(a)}
	\end{minipage}
	\begin{minipage}[htp]{0.25\textwidth}
		\centering
		\includegraphics[width=\textwidth]{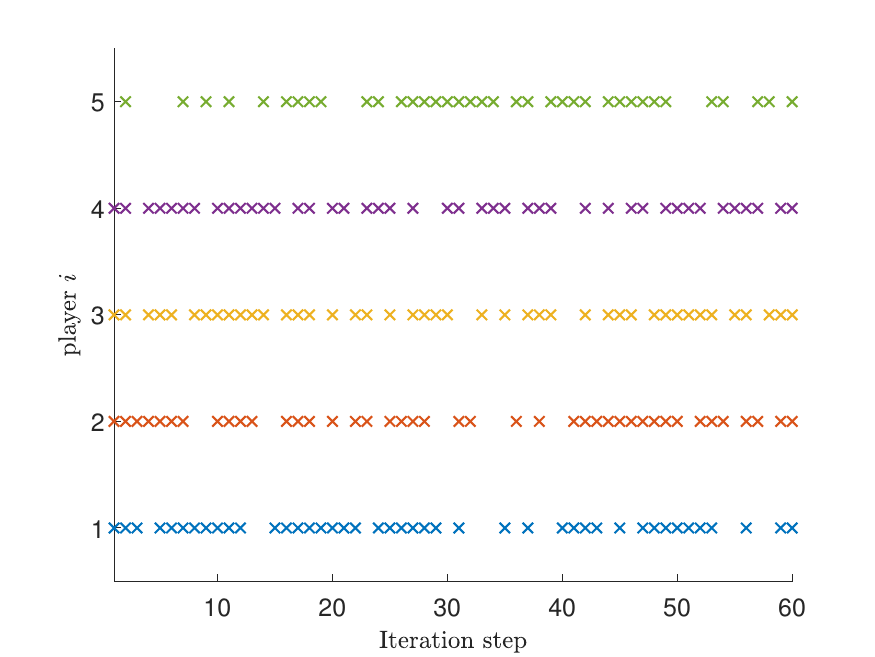}
		\centerline{(b)}
	\end{minipage}
	\begin{minipage}[htp]{0.25\textwidth}
		\centering
		\includegraphics[width=\textwidth]{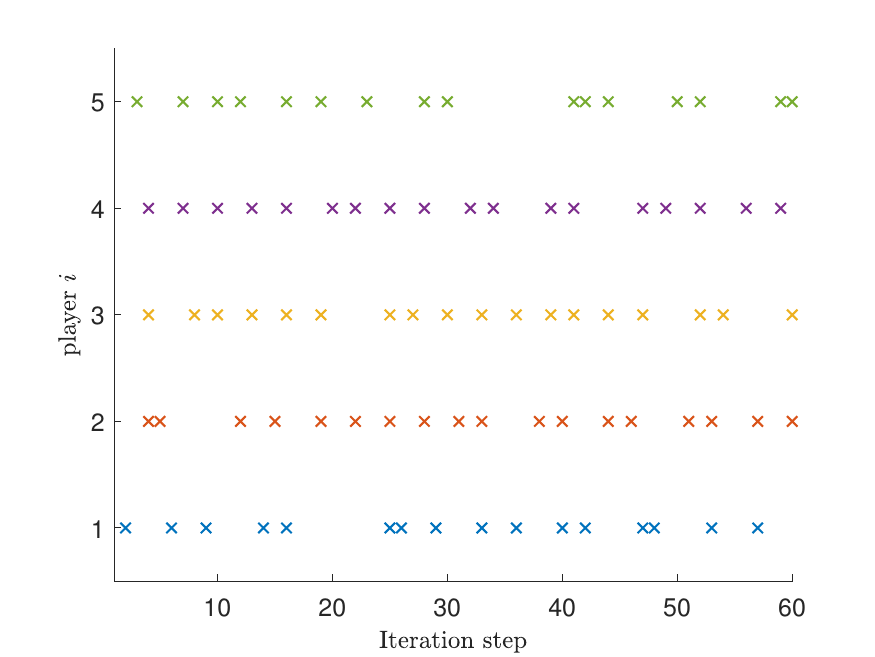}
		\centerline{(c)}
	\end{minipage}
	\begin{minipage}[htp]{0.25\textwidth}
		\centering
		\includegraphics[width=\textwidth]{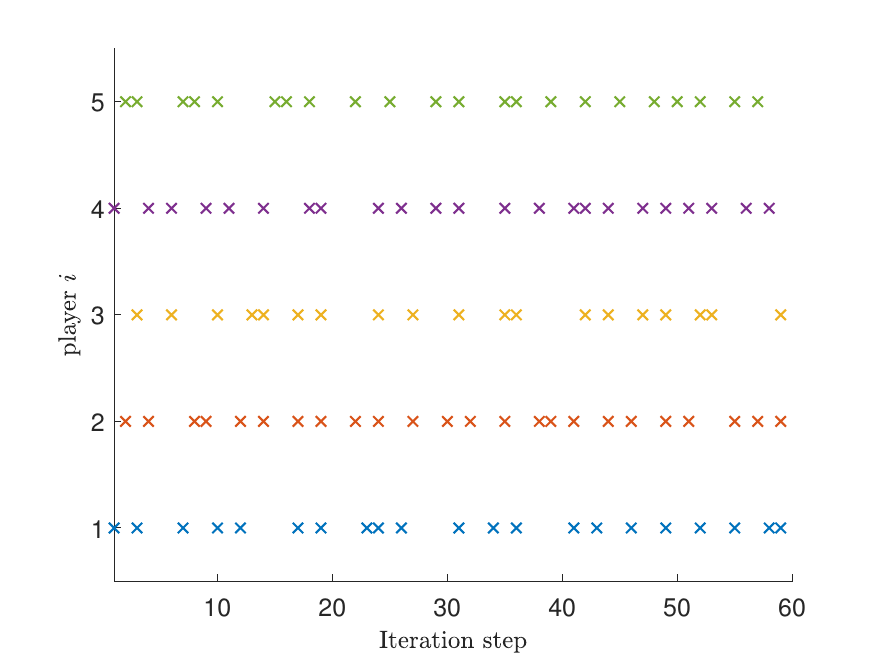}
		\centerline{(d)}
	\end{minipage}
	\begin{minipage}[htp]{0.25\textwidth}
		\centering
		\includegraphics[width=\textwidth]{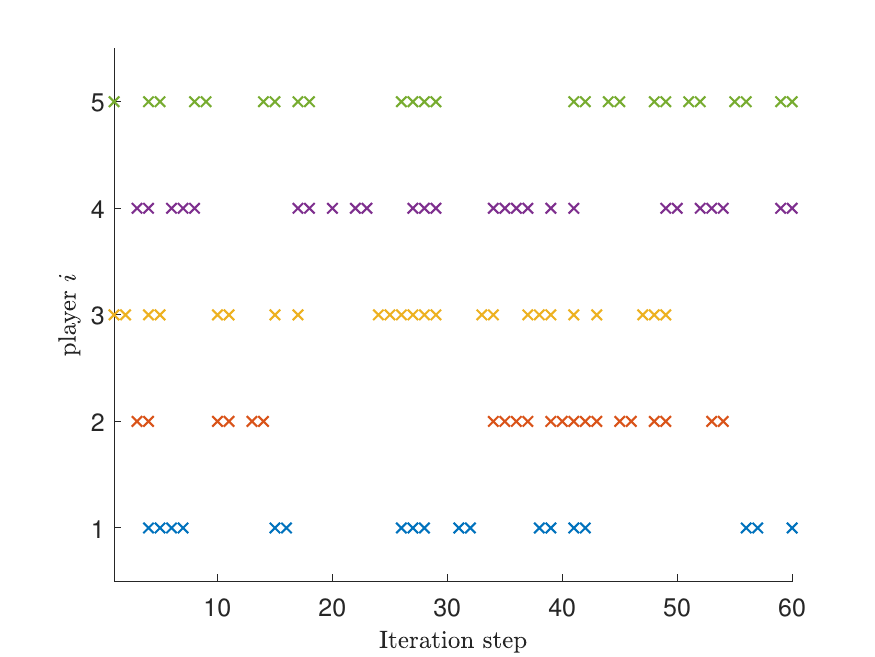}
		\centerline{(e)}
	\end{minipage}
	\begin{minipage}[htp]{0.25\textwidth}
		\centering
		\includegraphics[width=\textwidth]{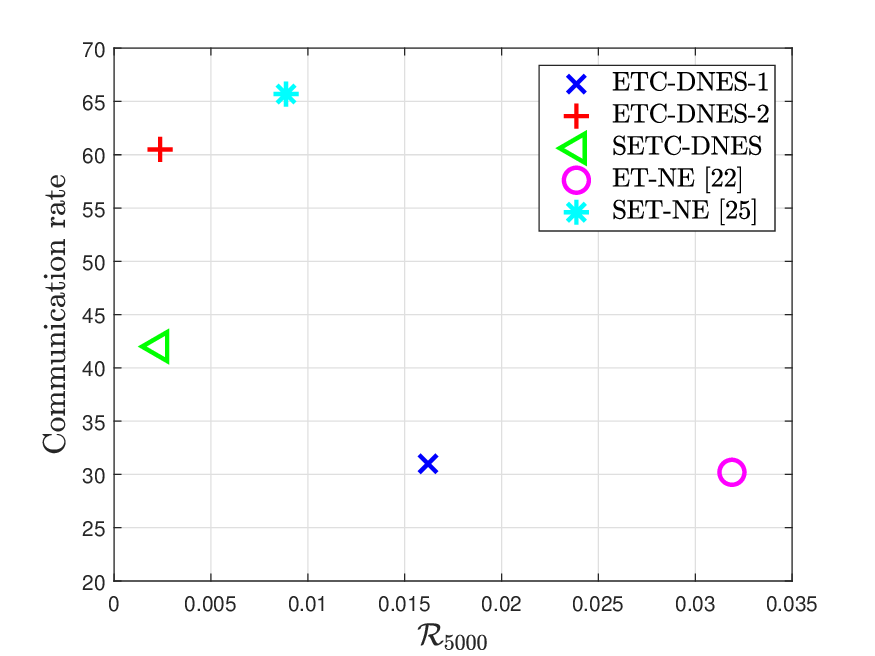}
		\centerline{(f)}
	\end{minipage}
	\caption{(a) Triggering instants for  ET-NE \cite{xu2021event}; (b) Triggering instants for SETC-NE \cite{huo2023distributed}; (c) Triggering instants for ETC-DNES-1; (d) Triggering instants for ETC-DNES-2; (e) Triggering instants for SETC-DNES; (f) A comparison on the residuals and communication rate (averaged trigger iterations/total iterations) by stopping the algorithms at $k = 5000$. }\label{instants}
	\centering
\end{figure}

\begin{figure}[htp]
	\centering
	\includegraphics[width=0.35\textwidth]{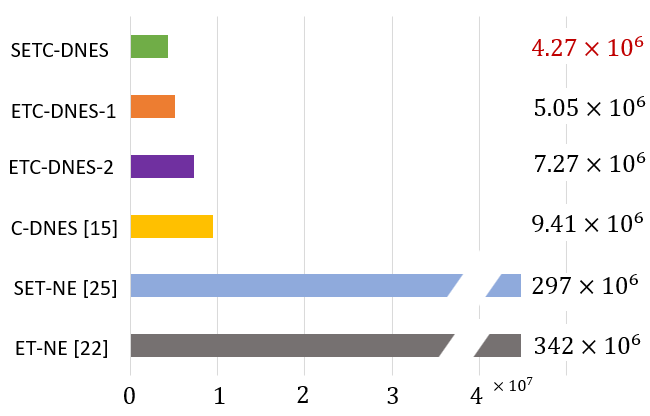}
	\caption{Total transmitted bits of different algorithms to obtain $\mathcal{R}\leq 0.01$.}
	\label{tot_comp}
\end{figure}

\textbf{Results.} We use the normalized residual $\mathcal{R}_k=\frac{||\mathbf{X}_k-\mathbf{X}^\star||_{\text{F}}}{||\mathbf{X}_0-\mathbf{X}^\star||_{\text{F}}}$ to measure the performance of each algorithm. We plot the convergence of the residual concerning both the number of iterations and the total bits transmitted for the various algorithm and triggering conditions combinations, all initialized with the same initial conditions. These results are shown in Fig. \ref{simu} (a) and (b), respectively. Moreover, the effectiveness of different triggering condition is illustrated in Fig. \ref{instants}   and a comparison of the total transmitted bits required by different algorithms to reach $\mathcal{R}\leq 0.01$ is provided in Fig. \ref{tot_comp}. Here are some key observations:

\begin{itemize}
	\item From Fig. \ref{simu} (a),  we can see that proposed communication-efficient algorithm with different triggering conditions all converge  to the NE   of the game. Moreover, ETC-DNES-2 and SETC-DNES  achieve  a similar  linear convergence rate as C-DNES without  event triggering \cite{chen2022linear}.
	\item   The effectiveness of different triggering conditions is depicted in Fig. \ref{instants} (a)--(e), clearly demonstrating a reduction in communication rounds. Regarding execution efficiency, a comprehensive comparison is provided in Fig. \ref{instants} (f). This comparison includes algorithm residuals and communication rates  for each algorithm at the time iteration $k=5000$. {\color{black} Notably, ETC-DNES-1 achieves a similar reduction in communication rounds as ET-NE \cite{xu2021event}, but with a smaller residual, indicating higher accuracy. For stochastic event-triggered algorithms, our proposed SETC-DNES   outperforms SET-NE \cite{huo2023distributed} in both the number of communication rounds and convergence performance, achieving both a low residual and a reduced communication rate.}
	\item From Fig. \ref{simu} (b), we observe that our proposed algorithms converge faster than C-DNES \cite{chen2022linear}, ET-NE\cite{xu2021event}, and SET-NE\cite{huo2023distributed} when comparing their performances based on the total number of communication bits, demonstrating the effectiveness of our proposed compressed algorithms. {\color{black} Moreover, Fig. \ref{tot_comp} shows that our compressed algorithms achieve reductions of up to a hundredfold in the total transmitted bits compared to ET-NE \cite{xu2021event} and SET-NE \cite{huo2023distributed}, which do not involve compression.} Moreover, SETC-DNES distinguishes itself by achieving the lowest communication cost among all the algorithms, utilizing only $1.24\%$ of the bits required by ET-NE\cite{xu2021event} and {\color{black}$1.42\%$ of those needed by SET-NE\cite{huo2023distributed}} to reach a specified accuracy of the algorithm.

	\end{itemize}

%

\section {CONCLUSION AND FUTURE WORK}
In this paper, we  have studied  the problem of distributed Nash equilibrium seeking over directed networks with constrained communication.  A novel compressed  event-triggered distributed NE seeking approach (ETC-DNES) is proposed to save communication, where a compressed information is  exchanged among neighboring agents only when an  event condition is satisfied. ETC-DNES can guarantee the exact convergence to the NE  using a diminishing triggering threshold sequence, and  linear convergence can be achieved under certain conditions.
To further improve the algorithm  
efficiency, we present a stochastic event-triggered mechanism for compressed NE seeking algorithm (SETC-DNES) and also prove its linear convergent property. Simulation examples demonstrate the effectiveness of the proposed communication-efficient algorithms for directed graphs, with SETC-DNES outperforming ETC-DNES in terms of total number of transmitted bits while achieving a small error. 	Future works may focus on the extensions to  time-varying  networks and  {\color{black}distributed generalized  NE seeking problems.   }


%

\section{APPENDIX}

\bibliographystyle{IEEEtran}
\bibliography{ref.bib}

\end{document}